\documentclass{article}

\usepackage{authblk}
\usepackage[inline]{enumitem}
\usepackage{fullpage}
\usepackage{amsfonts}
\usepackage{bm}
\usepackage{mathrsfs,amssymb}
\usepackage{nicefrac}
\usepackage{siunitx}
\usepackage{graphicx}
\usepackage{algorithm,algorithmic}
\usepackage[dvipsnames]{xcolor}
\usepackage{amsopn,amsmath,amsthm}
\usepackage{hyperref}
\usepackage{cleveref}
\usepackage[caption=false]{subfig}
\usepackage[export]{adjustbox}
\ifpdf
  \DeclareGraphicsExtensions{.eps,.pdf,.png,.jpg}
\else
  \DeclareGraphicsExtensions{.eps}
\fi

\hypersetup{
  frenchlinks=false,
  pdfborder={0 0 0},
  naturalnames=false,
  hypertexnames=false,
  breaklinks
}

\DeclareMathOperator{\Tr}{Tr}
\def\tr{\top}
\def\algcom#1{\hfill\textcolor{gray}{\textbackslash\textbackslash #1}}

\newtheorem{proposition}{Proposition}

\theoremstyle{definition}

\theoremstyle{remark}
\newtheorem{remark}{Remark}

\title{Redesigning the ensemble Kalman filter\\ with a dedicated model of epistemic uncertainty}

\author[1]{Chatchuea Kimchaiwong}
\author[2]{Jeremie Houssineau}
\author[3]{Adam M. Johansen}
\affil[1]{Warwick Mathematics Institute, University of Warwick, UK.}
\affil[2]{Division of Mathematical Sciences, Nanyang Technological University, Singapore.}
\affil[3]{Department of Statistics, University of Warwick, UK.}

\date{}

\newcommand{\ack}{\section*{Acknowledgements} AMJ acknowledges the financial support of the United Kingdom Engineering and Physical Sciences Research Council (EPSRC; grants EP/R034710/1 and EP/T004134/1) and by United Kingdom Research and Innovation (UKRI) via grant number EP/Y014650/1, as part of the ERC Synergy project OCEAN.}

\ifpdf
\hypersetup{
  pdftitle={Redesigning the ensemble Kalman filter with a dedicated model of epistemic uncertainty},
  pdfauthor={C. Kimchaiwong, J. Houssineau and A. M. Johansen}
}
\fi

\begin{document}

\maketitle

\begin{abstract}
The problem of incorporating information from observations received serially in time is widespread in the field of uncertainty quantification. Within a probabilistic framework, such problems can be addressed using standard filtering techniques. However, in many real-world problems, some (or all) of the uncertainty is epistemic, arising from a lack of knowledge, and is difficult to model probabilistically. This paper introduces a \emph{possibilistic} ensemble Kalman filter designed for this setting and characterizes some of its properties. Using possibility theory to describe epistemic uncertainty is appealing from a philosophical perspective, and it is easy to justify certain heuristics often employed in standard ensemble Kalman filters as principled approaches to capturing uncertainty within it.
The possibilistic approach motivates a robust mechanism for characterizing uncertainty which shows good performance with small sample sizes, and can outperform standard ensemble Kalman filters at given sample size, even when dealing with genuinely aleatoric uncertainty.
\end{abstract}

\paragraph{Keywords:}Bayesian inference; State-Space Model; Possibility Theory

\section{Introduction}

Dynamical real-world systems are typically represented as a hidden process in some state space, indirectly observed via a measurement process. The usual approach to characterizing the unknown state and uncertainty about it is to combine the observed information with given prior information in an approach known as filtering, or sometimes \emph{data assimilation} \cite{reich2015probabilistic}. However, this is non-trivial due not only to imprecise measurements, but also to model misspecification \cite{sarkka2023bayesian}.

Perhaps the best-known data-assimilation algorithm is the Kalman filter (KF) \cite{smc:hmm:Kal60}, which is the optimal filter for the linear Gaussian state-space model, as detailed, e.g., in \cite{sarkka2023bayesian}. However, its application is limited by its high computational costs when dealing with large state spaces, and by the fact that it cannot be applied directly to nonlinear models \cite{petrie2008localization}. Several modifications of the KF have been developed to extend its scope. The extended Kalman Filter can deal with a mild level of nonlinearity by approximating the model with a first order Taylor series expansion, as described, e.g., in \cite{terejanu2008extended}. The unscented Kalman Filter (UKF) approximates the distribution with Gaussian distribution using a set of sigma points and their weights, often giving better accuracy for highly nonlinear models \cite{wan2000unscented}. Still, the accuracy of the UKF depends upon the tuning of certain algorithmic parameters which determine the locations of the sigma points. Moreover, the number of sigma points (and hence the computational cost and accuracy) is fixed. Another widely used algorithm is the ensemble Kalman filter (EnKF) \cite{burgers1998analysis, tippett2003ensemble}, which uses an ensemble to approximate the distributions of interest. The strengths of the EnKF are that it does not require the calculation of the variance (avoiding costly matrix inversions in high-dimensional problems) and can be used with a nonlinear model without approximating it with its linear tangent model \cite{petrie2008localization} as one simply needs to be able to sample from the transition kernel of the dynamic model---although this comes at the cost of further approximating a non-Gaussian distribution with a Gaussian one. The EnKF must be used carefully with high-dimensional problems as constraints on computing power typically lead to an ensemble size relatively small compared to the dimension of the state space, leading to problems such as spurious correlation and underestimated variance, affecting the algorithm's performance. Thus, the recent development in this field mainly focuses on reducing computational cost and increasing the forecast's accuracy.

Although there has been much recent development in the field of data assimilation, e.g., \cite{arcucci2020neural, asch2016data, kalnay2003atmospheric}, the opportunities offered by alternative representations of the uncertainty have not been fully explored. Yet, there is a clear motivation for doing so: In many situations in which the EnKF is used, the main sources of uncertainty are \begin{enumerate*}[label=\roman*)]
    \item the initial state of the system,
    \item unknown deviations in the model, e.g., unknown forcing terms, and
    \item the lack of data.
\end{enumerate*}
These sources of uncertainty can be argued to be of an epistemic nature and, given their predominance, it is important to model them as faithfully as possible. This makes appealing the use of a dedicated framework for epistemic uncertainty. Many different approaches have been suggested in the literature to deal with epistemic uncertainty, such as fiducial inference \cite{fisher1935fiducial, hannig2016generalized} and Dempster--Shafer theory \cite{dempster2008dempster}. While fiducial inference gives a probabilistic statement about a parameter without relying on a prior probability distribution, it still represents the information a posteriori via probability distributions. On the other hand, Dempster--Shafer theory leverages (fuzzy) set-valued probabilities to represent both aleatoric and epistemic uncertainty \cite{pearl1990reasoning}. Yet, its formulation, based on sums over all subsets of the state space, restricts its application to discrete cases and, even then, does not scale easily to larger problems. Another approach, termed possibility theory \cite{de1999conditioning, dubois2015possibility, dubois2000possibility}, proposes to model exclusively epistemic uncertainty by changing the algebra underpinning probability theory, i.e., summation/integration is replaced by maximisation. Many probabilistic concepts and analyses can be adapted to possibility theory, for instance, \cite{houssineau2019elements} proves a Bernstein-von Mises theorem for a possibilistic version of Bayesian inference. In real systems, some uncertainty will be well modelled as aleatoric and some will of an epistemic nature; \cite{houssineau2018parameter} shows that probability theory and possibility theory can be combined in this case, offering an extension of the probabilistic approach rather than an alternative. Possibility theory can be applied to complex problems, for instance \cite{houssineau2021linear} introduces an analogue to spatial point processes, and \cite{houssineau2018smoothing} shows that the Kalman filter can also be derived in this framework. These results highlight the potential of possibility theory for data assimilation, which we explore in this work via a possibilistic analogue of the EnKF. Related work includes \cite{bishop2023robust} in which a robust version of the Kalman filter is proposed based on sets of probability distributions, with a motivation close to the one for this work. 

Our contributions are as follows:
\begin{enumerate}
    \item We introduce a new notion of Gaussian fitting for possibility theory, and contrast it against the standard moment-matching procedure. We highlight the potential of this approach for providing principle grounds for standard heuristics used in the EnKF.
    \item We derive a possibilistic analogue of the EnKF, adapting the initialisation and predictions steps to the assumed epistemic uncertainty in the model, and showing how the update step of existing versions of the EnKF remain valid in this context.
    \item We assess the performance of our method against two versions of the EnKF and against the UKF, considering both linear and nonlinear dynamics as well as fully and partially observed processes.
\end{enumerate}

The remainder of this paper is organised as follows. \Cref{sec: Background} starts with the state estimation problem for the state-space model before giving the filtering techniques used to estimate such states under the probabilistic framework. Then, the details of the possibilistic framework developed to tackle epistemic uncertainty are given, with some concepts defined analogously to the probabilistic framework. Next, the possibilistic EnKF is introduced in \Cref{sec: Possibilistic EnKF} to estimate the system's state in the presence of epistemic uncertainty. After that, the performance of our algorithm is assessed on simulated data in a range of situations and compared against standard baselines in \Cref{sec: Results}.

\section{Background}
\label{sec: Background}

We first briefly review the standard EnKF and the underlying filtering problem in \Cref{subsec: State Estimation Problem and Filtering Technique} before introducing, in \Cref{subsec: Possibilistic Framework}, the framework based on which the proposed method will be derived.

\subsection{State estimation problem and filtering technique}
\label{subsec: State Estimation Problem and Filtering Technique}

The standard, probabilistic Gaussian state-space model describes the evolution of a hidden \emph{state} process, $(X_k)_{k > 0}$, and its associated observation process, $(Y_k)_{k>0}$, via
\begin{subequations}
\label{Gaussian State Space equations}
\begin{align}
\label{eq:SSMprediction}
X_k & = F_k(X_{k-1}) +\epsilon_k &\text{where}  &\quad \epsilon_k \sim \mathrm{N}(0,U_k) \\
\label{eq:SSMobservation}
Y_k & = H_k(X_k) +\varepsilon_k &\text{where}  &\quad \varepsilon_k \sim \mathrm{N}(0,V_k),
\end{align}
\end{subequations}
with $U_k$ and $V_k$ positive definite matrices and with $\mathrm{N}(\mu,\Sigma)$ denoting a Gaussian distribution of mean $\mu$ and covariance $\Sigma$.

We consider the setting in which a stream of observations $y_1,y_2,\dots$ arrives sequentially in time, and we want to obtain estimates of the hidden states as they arrive. Thus, in the Bayesian context, we characterise our knowledge of the state $X_k$ at time step $k > 0$ given the realisations $y_{1:k} \doteq (y_1,y_2,\dots,y_k)$ of the observations $Y_{1:k}$ via the filtering density $p(x_{k} | y_{1:k})$ which can be computed recursively, see, e.g., \cite{sarkka2023bayesian}, as
\begin{align*}
  \tag{prediction}
  p(x_{k} | y_{1:k-1}) & = 
\int p(x_{k} | x_{k-1})p(x_{k-1} |y_{1:k-1}) \mathrm{d} x_{k-1}
\\
\tag{update}
p(x_{k} | y_{1:k}) & = 
\frac{p(y_k | x_k)p(x_k | y_{1:k-1})}{\int p(y_k | x_k)p(x_k | y_{1:k-1}) \mathrm{d} x_k}. 
\end{align*}

The Kalman Filter (KF) provides a closed-form recursion for the prediction and update steps for linear Gaussian state-space models (see, e.g., \cite{sarkka2023bayesian}). However, many models in this world are nonlinear, meaning that the KF is not directly applicable and, unfortunately, there are few other settings in which analytic recursions are available. This motivates the development of numerous alternatives, including the UKF and the EnKF.

While the UKF can provide good performance for moderately nonlinear models, which depends upon careful choice of several algorithmic tuning parameters, and, as the number of sigma points is fixed, there is little scope for adjusting computational cost or accuracy. The EnKF algorithm takes a different approach; it can be viewed as an approximation to the KF which employs an ensemble of samples $\{X_k^i\}_{i=1}^N$ to approximate the parameters of the Gaussian distribution at time $k$. In particular, the ensemble's mean and variance are used to approximate those of the filtering distribution.
The EnKF algorithm then proceeds recursively via the prediction and update steps. The prediction step mirrors that of the KF but uses an ensemble approximation, which means that $F_k$ need not be linear as long as it can be evaluated pointwise as in \Cref{alg: prediction step of EnKF} \cite{burgers1998analysis, tippett2003ensemble}.

\begin{algorithm}
\caption{Prediction step of EnKF at time step $k$}
\label{alg: prediction step of EnKF}
	\renewcommand{\algorithmicrequire}{\textbf{Input:}}
	\renewcommand{\algorithmicensure}{\textbf{Output:}}
	\begin{algorithmic}[1]
		\REQUIRE The posterior ensemble $\{\hat{X}_{k-1}^i\}_{i=1}^N$ at time step $k-1$.
		\ENSURE  The predictive ensemble $\{X_{k}^i\}_{i=1}^N$ at time step $k$ with mean $\mu_k$ and variance $\Sigma_k$.\\
        \FOR{$i \in \{1, \ldots, N \}$} 
        \STATE $\epsilon_{k}^i \sim \mathrm{N}(0,U_k)$ \algcom{Sample noise}
		\STATE $X_{k}^i = F_k(\hat{X}_{k-1}^i) + \epsilon_{k}^i$  \algcom{Predict $i$-th particle}
        \ENDFOR
        \STATE $\mu_k \leftarrow \frac{1}{N}\sum_{i=1}^N X_k^i$ \algcom{Predictive mean}
        \STATE $\Sigma_k \leftarrow \frac{1}{N-1}\sum_{i=1}^N (X_k^i - \mu_k)(X_k^i - \mu_k)^{\tr}$, \algcom{Predictive variance}
\end{algorithmic}
\end{algorithm}

The update step of the EnKF also approximates that of the KF: the mean of the updated distribution is a convex combination of that of the predictive distribution and the observation using the so-called \emph{Kalman gain} as a weight. However, using the traditional Kalman gain can lead to an underestimation in the posterior variance (as it essentially neglects the observation noise) and can lead to filter divergence \cite{whitaker2002ensemble}. A number of variants of the EnKF update step have been developed to address this problem, including both the stochastic EnKF (StEnKF) and the square root EnKF (SqrtEnKF) \cite{katzfuss2016understanding}. The latter updates the deviation of the ensemble by applying a matrix obtained from the square root of a certain matrix. Different variations of the SqrtEnKF exist, depending on how such a matrix is obtained, as summarised in \cite{tippett2003ensemble}. In this work, we consider the version of the SqrtEnKF presented in \cite{whitaker2002ensemble}, which is particularly suitable as a basis for introducing a possibilistic version of the EnKF. The update of the SqrtEnKF for a linear observation model is detailed in \Cref{alg: update step of SqrtEnKF}, where the notation $A^{\nicefrac{1}{2}}$ refers to the Cholesky factor of matrix $A$, and where we use the same notation for a linear function and for its matrix representation.

\begin{algorithm}
\caption{Update step of SqrtEnKF at time step $k$}
\label{alg: update step of SqrtEnKF}
	\renewcommand{\algorithmicrequire}{\textbf{Input:}}
	\renewcommand{\algorithmicensure}{\textbf{Output:}}
	\begin{algorithmic}[1]
		\REQUIRE The predictive ensemble $\{X_{k}^i\}_{i=1}^N$ at time step $k$ with mean $\mu_k$ and variance $\Sigma_k$.
		\ENSURE  The posterior ensemble $\{\hat{X}_{k}^i\}_{i=1}^N$ at time step $k$.
		\STATE $S_k \leftarrow H_k\Sigma_{k}H_k^{\tr} + V_k$ \algcom{Covariance of the innovation from the ensemble}
            \STATE $K_k \leftarrow \Sigma_{k}H_k^{\tr} S_{k}^{-1}$ \algcom{Standard Kalman gain}
            \STATE $\tilde{K}_k \leftarrow \Sigma_kH_k^{\tr}\big(S_{k}^{\nicefrac{1}{2}}\big)^{-\tr}\big( S_{k}^{\nicefrac{1}{2}} + V_k^{\nicefrac{1}{2}} \big)^{-1}$ \algcom{Adjusted Kalman gain}
            \STATE $\hat{\mu}_{k} \leftarrow \mu_k + K_k(Y_k - H_k \mu_k)$ \algcom{Posterior mean}
            \FOR{$i \in \{1, \ldots, N \}$}
            \STATE $\hat{e}^{i}_{k} \leftarrow (I_{n} -\tilde{K}_kH_k)(X_{k}^i - \mu_{k})$ \algcom{Posterior deviation}
            \STATE $\hat{X}_{k}^i \leftarrow \hat{e}^{i}_{k} + \hat{\mu}_{k}$ \algcom{Updated particle}
            \ENDFOR
\end{algorithmic}
\end{algorithm}

For the ensemble-based method, the accuracy of the state estimate depends highly on the ensemble size. The situation in which the ensemble is too small to adequately represent the system is termed undersampling and leads to four main problems which have been termed: inbreeding, spurious correlation, matrix singularity, and filter divergence \cite{petrie2008localization}. Two widely used techniques to negate the undersampling problems are inflation and localisation \cite{petrie2008localization}. The inflation technique increases the deviation between the sample and the predictive mean so that the predictive covariance from the ensemble is not underestimated and avoids the so-called inbreeding problem. On the other hand, the localisation techniques are implemented to eradicate the spurious correlation and singular matrix, examples being tapering using the Schur product to cut off the long-range correlation or domain localisation where assimilation is performed in a local domain, which is disjoint in a physical space \cite{janjic2011domain}. Each technique needs to be fine-tuned in a problem-specific manner to achieve good accuracy.

We have focused here on approaches directly relevant to the methodology developed below in \Cref{sec: Possibilistic EnKF}; many other methods have been devised to address nonlinearity and non-Gaussianity---see, e.g., \cite{asch2016data, fearnhead2018particle, van2015nonlinear}. 

\subsection{Review of possibility theory}
\label{subsec: Possibilistic Framework}

Possibility theory aims to directly capture epistemic uncertainty about a fixed but unknown element $\theta^*$ in a set $\Theta$, by focusing on the \emph{possibility} of an event related to $\theta^*$ rather than on its probability. A possibility of $1$ corresponds to the absence of evidence against the event taking place \emph{not} evidence that the event took place almost surely. If there is no evidence against an event, say $E$, then there cannot be any evidence against ($E$ or $E'$), with $E'$ any other event, so that the possibility of ($E$ or $E'$) is also equal to $1$. This behaviour shows the notion of possibility does not give an additive measure of uncertainty, and the simplest operation that corresponds to the possibility of a union of events is the \emph{maximum} of their individual possibilities.

The analogue of a probability density is a so-called possibility function, $f:\Theta \to [0,1]$ which must have supremum $1$. Just as probabilities are best described by measures, possibilities are naturally cast as outer measures, $\bar{P}$ and just as one may obtain a probability from  a density one can obtain a possibility from a possibility function by taking its supremum, that is, for any $A \subseteq \Theta$, $\bar{P}(A)=\sup_{\theta \in A}f(\theta)$. If it holds that $f(\theta) = 1$ for some $\theta \in A$, then we will indeed have $\bar{P}(A) = \bar{P}(A \cup B) = 1$ for any $B \subseteq \Theta$, as required. Formally, the set function $\bar{P}$ is an outer measure verifying $\bar{P}(\Theta) = 1$; we therefore refer to such set functions as \emph{outer probability measures} (o.p.m.s). If $\bar{P}(A) = 1$ models the absence of opposing evidence, $\bar{P}(A) = 0$ has essentially the same interpretation as with probabilities, i.e., that $\theta^* \in A$ is impossible, which requires very strong evidence. In general, $\bar{P}(A)$ can be interpreted as an upper bound for \emph{subjective} probabilities of the event $\theta^* \in A$, i.e., the maximum probability that we would be ready to assign to this event is $\bar{P}(A)$. Under this interpretation, the statement $\bar{P}(A) = 1$ is uninformative: we could assign any probability in the interval $[0,1]$ to this event.

To formally define the quantities described above, we consider an analogue of a random variable, referred to as an \emph{uncertain variable}, defined as a mapping $\bm{\theta} : \Omega \to \Theta$, and interpreted as follows: if the true outcome in $\Omega$ is $\omega$ then $\bm{\theta}(\omega)$ is the true value of the parameter. The set $\Omega$ plays a similar role to the elementary event space in probability theory, except that it is not equipped with any probabilistic structure; instead, it contains a true element $\omega^*$ and it holds by construction that $\bm{\theta}(\omega^*) = \theta^*$. An event can now be defined as a subset of $\Omega$, for instance $\{\omega \in \Omega : \bm{\theta}(\omega) \in A\}$ for some $A \subseteq \Theta$. We will use the standard shortcut and write this event as $\bm{\theta} \in A$. An important difference between possibility theory and probability theory is that uncertain variables do not characterise possibility functions, instead possibility functions \emph{describe} uncertain variables in a way that is not unique. For instance, if the available information is modelled by a possibility function $f$ describing an uncertain variable $\bm{\theta}$, then any function $g$ such that $g \geq f$, i.e.\ $g(\theta) \geq f(\theta)$ for any $\theta \in \Theta$, also describes $\bm{\theta}$. This is because $g$ discarded some evidence about $\bm{\theta}$, that is, $g$ is less informative than $f$. In particular, we can consider the possibility function equal to $1$ everywhere, denoted $\bm{1}$, which upper bounds all other possibility functions. The function $\bm{1}$ models the total absence of information.

To gain information from data, Bayesian inference can be performed within possibility theory in a way that is very similar to the standard approach: If $Y$ is a random variable with probability distribution $p(\cdot \mid \theta^*)$ belonging to a parameterised family of distributions\footnote{Although $p(\cdot \mid \theta)$ is not formally a conditional probability distribution, it is useful to slightly abuse notations and write it as such.} $\{p(\cdot \mid \theta) : \theta \in \Theta\}$, if we observe a realisation $y$ of $Y$, and if the information available about $\theta^*$ a priori is modelled by the possibility function $f$, then the information available a posteriori can be modelled by the possibility function $f(\cdot \mid y)$ characterised by \cite{houssineau2018parameter}
\begin{equation}
\label{eq:updateByRandomObs}
f(\theta \mid y) = \frac{p(y \mid \theta)f(\theta)}{\sup_{\theta' \in \Theta} p(y \mid \theta') f(\theta')}, \qquad \theta \in \Theta.
\end{equation}
We borrow from the Bayesian nomenclature and simply refer to $f$ and $f(\cdot \mid y)$ as the prior and the posterior. Since we can always start from the uninformative prior $f = \bm{1}$, it is easy to find posterior possibility functions by inserting different likelihoods; for instance, if $\Theta = \mathbb{R}^n$ and if the likelihood is a multivariate Gaussian distribution with mean $\theta$ and known variance $\Sigma$, i.e., $p(y \mid \theta) = \mathrm{N}(y; \theta, \Sigma)$, then
\[
f(\theta \mid y) = \exp\bigg( -\frac{1}{2}(\theta - y)^{\tr}\Sigma^{-1}(\theta - y) \bigg) \doteq \overline{\mathrm{N}}(\theta; y, \Sigma).
\]
Such a Gaussian possibility function is a conjugate prior for the Gaussian likelihood as in the probabilistic case, and shares many of the properties of its probabilistic analogue. It can be advantageous to parameterise a Gaussian possibility function, say $\overline{\mathrm{N}}(\theta; \mu, \Sigma)$, by the precision matrix $\Lambda = \Sigma^{-1}$; indeed, the precision matrix does not need to be positive definite for the Gaussian possibility function to be well defined, with positive semi-definiteness being sufficient. In particular, this means that setting $\Lambda$ to the $0$ matrix is possible, with the Gaussian possibility function being equal to $\bm{1}$ in this case. A simple way to quantify the amount of epistemic uncertainty in a possibility function is to consider the integral $\int f(\theta) \mathrm{d} \theta$ as in \cite{chen2021observer}, when defined. This notion of uncertainty is consistent with the partial order on possibility functions: if $g$ is less informative than $f$ then the integral of $g$ will obviously be larger then that of $f$. In particular, the uncertainty of a Gaussian possibility function is $\sqrt{|2 \pi \Sigma|}$, with $|\cdot|$ denoting the determinant; a quantity often used to quantify how informative a given Gaussian distribution is.

If $\bm{\theta}$ and $\bm{\psi}$ are two uncertain variables, respectively on sets $\Theta$ and $\Psi$, jointly described by a possibility function $f$ on $\Theta \times \Psi$, then Bayes' rule can also be expressed via a more standard notion of conditioning as \cite{de1999conditioning}
\[
f(\theta \mid \psi) = \frac{f(\theta,\psi)}{\sup_{\theta' \in \Theta} f(\theta', \psi)},
\]
where $\sup_{\theta' \in \Theta} f(\theta', \psi)$ is the marginal possibility function describing $\bm{\psi}$. The conditional possibility function $f(\theta \mid \psi)$ allows independence to be defined simply: if $f(\theta \mid \psi)$ is equal to the marginal $f(\theta) = \sup_{\psi' \in \Psi} f(\theta, \psi')$ for any $\theta \in \Theta$, then $\bm{\theta}$ is said to be independent from $\bm{\psi}$ under $f$. As with most concepts in possibility theory, the notion of independence depends on the choice of possibility function; if $\bm{\theta}$ and $\bm{\psi}$ are not independent under $f$, we could find another possibility function $g$ such that $g \geq f$ and such that $\bm{\theta}$ is independent of $\bm{\psi}$ under $g$. This will be key later on for simplifying high-dimensional Gaussian possibility functions in a principled way. When there is no ambiguity about the underlying possibility function, we will simply say that $\bm{\theta}$ is independent of $\bm{\psi}$.

Although the parameters $\mu$ and $\Sigma$ of the Gaussian possibility functions are reminiscent of the probabilistic notions of expected value and variance, these notions have to be redefined in the context of possibility theory. Asymptotic considerations \cite{houssineau2019elements} lead to a notion of expected value $\mathbb{E}^*(\cdot)$ based on the mode, that is
\(
\mathbb{E}^*(\bm{\theta}) = \arg\max_{\theta \in \Theta} f(\theta),
\)
for any uncertain variable $\bm{\theta}$ described by $f$, and to a local notion of variance based on the curvature of $f$ at $\mathbb{E}^*(\bm{\theta})$, assuming that the latter is a singleton, an assumption that we will make throughout this work. These quantities correspond to the approximation often referred to as Laplace/Gaussian approximation, and make sense in the context of epistemic uncertainty: without a notion of variability, we simply look at our best guess $\mathbb{E}^*(\bm{\theta})$ for the unknown $\theta^*$ and at how confident we are in that guess, i.e., how fast the possibility decreases around it. Although the expected value and variance associated with the Gaussian possibility function $\overline{\mathrm{N}}(\mu, \Sigma)$ are indeed $\mu$ and $\Sigma$, they will differ in general from their probabilistic counterparts.

In order to make inference for dynamical systems simpler, we introduce a Markovian structure as follows:  We consider a sequence of uncertain variables $\bm{x}_0, \bm{x}_1, \dots$ in a set $\mathsf{X}$ and assume that $\bm{x}_k$ is described by a possibility function $f_k$ for some given $k \geq 0$. Following the standard approach, we assume that $\bm{x}_{k+1}$ is conditionally independent of $\bm{x}_{k-\delta}$ given $\bm{x}_k$, for any $\delta > 0$. The available information about $\bm{x}_{k+1}$ given a value $x_k$ of $\bm{x}_k$ can then be modelled by a conditional possibility function $f_{k+1|k}( \cdot \mid x_k)$, and prediction can be performed via
\[
f_{k+1}(x_{k+1}) = \sup_{x_k \in \mathsf{X}} f_{k+1|k}(x_{k+1} \mid x_k) f_k(x_k), \qquad x_{k+1} \in \mathsf{X}.
\]
We will focus in particular on the case where an analogue of \eqref{eq:SSMprediction} holds, that is
\[
\bm{x}_k = F_k(\bm{x}_{k-1}) + \bm{u}_k,
\]
with $\bm{u}_k$ an uncertain variable described by $\overline{\mathrm{N}}(0, U_k)$. In this case, the term $\bm{u}_k$ cannot be interpreted as noise and models instead deviations between the model and the true dynamics: the true states $x^*_k$ and $x^*_{k-1}$ are related via $x^*_k = \tilde{F}_k(x^*_{k-1})$ with $\tilde{F}_k$ potentially different from $F_k$; the true value $u^*_k$ of interest is therefore equal to the difference $\tilde{F}_k(x^*_{k-1}) - F_k(x^*_{k-1})$ between the model and the true dynamics.

In the case where the prediction is deterministic, that is $\bm{x}_{k+1} = F_k(\bm{x}_k)$ for some possibly non-linear mapping $F_k$, we can use the change of variable formula \cite{baudrit2008representing} to characterise the possibility function $f_{k+1}$ as
\begin{equation}
f_{k+1}(x_{k+1}) = \sup\big\{ f_k(x_k) : x_k \in F_k^{-1}(x_{k+1}) \big\},
\label{change of variable equation}
\end{equation}
where $F_k^{-1}(x_{k+1})$ is the (possibly set-valued) pre-image of $x_{k+1}$ via $F_k$, and $\sup \emptyset = 0$ by convention. A result that is specific to possibility theory is that the expected values of $\bm{x}_{k+1}$ and $\bm{x}_k$ are related via $\mathbb{E}^*(\bm{x}_{k+1}) = F_k(\mathbb{E}^*(\bm{x}_k))$ without assumptions on $F_k$. This result will be key in our approach since it allows to compute the expected value at time step $k+1$ with a single application of $F_k$, rather than averaging ensembles as in the EnKF. This result does not hold for the mode of probability densities because the corresponding change of variable formula include the Jacobian of the transformation, which shifts the mode in non-trivial ways.

\section{Possibilistic EnKF}
\label{sec: Possibilistic EnKF}

It is has been shown in \cite{houssineau2018smoothing} that an analogue of the Kalman filter can be derived in the context of possibility theory, and that the corresponding expected value and variance are the same as in the probabilistic case. However, as with the probabilistic KF, its applicability is limited to linear models---so it is natural to develop extensions that accommodate nonlinearity. Since the probabilistic EnKF is flexible in computational cost and can be effectively applied to nonlinear models, we explore the use of EnKF-like ideas in the possibilistic setting. Here, we present a novel \emph{possibilistic EnKF} (p-EnKF) and show that analogues of inflation and localisation arise naturally in the context of possibility theory rather than being imposed as heuristics to improve performance as in the standard setting.

For the p-EnKF, we use an ensemble of weighted particles to characterise the possibility function and follow a similar path to the standard EnKF by assuming that the underlying possibility function is Gaussian in order to proceed with a KF-like update. For this purpose, we need to define two operations:
\begin{enumerate*}[label=\arabic*)]
    \item how to approximate a given possibility function by weighted particles? This will be necessary for initialisation at time step $k=0$, and
    \item how to define a Gaussian possibility function based on weighted particles? This will be necessary to carry out the Kalman-like update.
\end{enumerate*}
We start by answering the latter question in \Cref{mean and variance poss EnKF}, before moving to the former in \Cref{Ini poss EnKF}. Based on this, we detail the prediction and update mechanism of the p-EnKF in \Cref{Pre poss EnKF,Up poss EnKF} respectively, and consider some extensions in \Cref{sec:extensions}.

\subsection{Ensemble approximation}
\label{sec:ensembleApproximation}

We consider the problem of defining a set $\{(w_i,x_i)\}_{i=1}^N$ of $N$ weighted particles in $\mathbb{R}^n$, which will allow the approximate solution of optimisation problems of the form $\max_{x \in \mathbb{R}^n} \varphi(x)f(x)$, for some bounded function $\varphi$ on $\mathbb{R}^n$ and for a given possibility function on $\mathbb{R}^n$. Although our approach could be easily formulated on more general spaces than $\mathbb{R}^n$, Euclidean spaces are sufficient for our purpose in this work and help to simplify the presentation. As is common in Monte Carlo approaches, we do not want the particles or weights to depend on $\varphi$, so as to allow us to solve this problem for many different such functions using a single sample, we therefore focus on directly approximating $f$: placing particles in locations which allow a good characterization of $f$ allows us to approximate the optimization for a broad class of regular $\varphi$. Following the principles of Monte Carlo optimization \cite[Chapter 5]{robert1999monte}, and assuming that $\int f(x) \mathrm{d} x < +\infty$, we can sample from the probability distribution proportional to $f$ to obtain the particles $\{x_i\}_{i=1}^N$ and then weight these particles with $w_i = f(x_i)$, for any $i \in \{1,\dots,N\}$. We then obtain the approximation
\[
\max_{x \in \mathbb{R}^n} \varphi(x)f(x) \approx \max_{i \in \{1,\dots,N\}} w_i \varphi(x_i).
\]
This approximation can be proved to converge when $N \to \infty$ under mild regularity conditions on $\varphi$ as long as the supports of $\varphi$ and $f$ are not disjoint. This is a simple default choice for the methods developed below, although more sophisticated approaches are possible.

\begin{remark}
One initially appealing idea is to consider a form of maximum entropy principle by identifying the constraints we want to impose on the probability distributions from which we may wish to sample, and to pick the distribution of maximum entropy subject to satisfying these constraints. This is a standard argument for specifying the ``least informative'' distribution with particular properties, e.g., given a distribution when the mean and variance the maximum entropy distribution is the Gaussian with these moments. In the considered context, the constraints comes from the interpretation of o.p.m.s as upper bounds for probability distributions. Specifically, for an o.p.m.\ $\bar{P}$, we could draw an ensemble from the probability density $p$ satisfying the following criteria: $p$ has the maximum entropy amongst those distributions satisfying $\int_A p(x) dx \leq \bar{P}(A)$ for any $A \subseteq \mathbb{R}^n$. The distribution $p$ would typically be more diffuse than the one proportional to $f$, which can be beneficial for optimizing functions $\varphi$ taking large values in the ``tails'' of $f$. However, obtaining such a distribution is in general highly non-trivial beyond discrete or univariate problem.
\end{remark}

As there are no constraints on the way in which particle locations, $\{x_i\}_{i=1}^N$ are selected, deterministic schemes such as quasi Monte Carlo \cite{lemieux2009monte, guth2021quasi} could also be considered. In \Cref{sec: Results}, we will consider using the same approach as in the UKF to define particle locations, which will prove to be beneficial in some situations.

In the situations that we will consider, it will always be the case that $f$ is maximised at a single known element $x^* \in \mathbb{R}^n$. Therefore, it makes sense to add one particle $x_0 = x^*$ with weight $w_0 = 1$. This will prove beneficial when fitting a Gaussian possibility function to the ensemble $\{(w_i,x_i)\}_{i=0}^N$, as detailed in the next section.

\subsection{Best-fitting Gaussian possibility function} \label{mean and variance poss EnKF}

We consider the situation in which the epistemic uncertainty is captured by a set $\{(w_i,x_i)\}_{i=0}^N$ of $N+1$ weighted particles in $\mathbb{R}^n$, with $w_i = 1$ if and only if $i = 0$. We denote by $\tilde{f}$ the ``empirical'' possibility function defined based on the ensemble as $\tilde{f}(x) = \max_{i \in \{0,\dots,N\}} w_i \bm{1}_{x_i}(x)$, 
with $\bm{1}_{x_i}$ the indicator of the point $x_i$. The standard approach would be to simply compute the (weighted) mean and variance of the ensemble, but these notions do not apply directly to possibility functions, and the curvature-based possibilistic notion of variance is not defined for an empirical possibility function like $\tilde{f}$. Instead, we aim to fit a Gaussian possibility function $\overline{\mathrm{N}}(\mu, \Sigma)$ to this (weighted) ensemble and the considered variance will simply be the second parameter $\Sigma$ of this fitted Gaussian.

To fit the Gaussian $\overline{\mathrm{N}}(\mu, \Sigma)$ to the ensemble, we need a notion of best fit. Based on the partial order between possibility functions described in \Cref{subsec: Possibilistic Framework}, we can easily ensure that $\overline{\mathrm{N}}(\mu, \Sigma)$ does not introduce artificial information---at least at locations $\{x_i\}_{i=0}^N$ -- by requiring that $\overline{\mathrm{N}}(\mu, \Sigma) \geq \tilde{f}$. Since $\tilde{f}(x) = 0$ when $x \notin \{x_i\}_{i=0}^N$, we can simplify this condition to: $\overline{\mathrm{N}}(x_i; \mu, \Sigma) \geq w_i$ for all $i \in \{0,\dots,N\}$. The inequality $\overline{\mathrm{N}}(\mu, \Sigma) \geq \tilde{f}$ forces the expected values of $\overline{\mathrm{N}}(\mu, \Sigma)$ and $\tilde{f}$ to coincide, from which we can deduce that $\mu = x_0$.

For the variance, we aim to minimise the uncertainty in the possibility function $\overline{\mathrm{N}}(\mu, \Sigma)$, i.e., minimising $\int \overline{\mathrm{N}}(x; \mu, \Sigma) \mathrm{d} x \propto \sqrt{|\Sigma|}$, which is equivalent to maximising $\log |\Lambda|$ with $\Lambda = \Sigma^{-1}$ the precision matrix, thanks to the properties of the determinant. The precision matrix is then most naturally defined as the one solving the constrained optimisation problem
\begin{align}
\label{eq:secondOptimisation}
 \max_{\Lambda \in \mathbf{S}^d_+} \log |\Lambda|  &
& \text{subject to} \quad  \overline{\mathrm{N}}(x_i; \mu, \Sigma) \geq w_i, \quad 1 \leq i \leq N,
\end{align}
where $\mathbf{S}^d_+$ is the cone of positive semi-definite $d \times d$ matrices. Since $\mu = x_0$, the corresponding constraint $\overline{\mathrm{N}}(x_0; \mu, \Sigma) \geq w_0$ is automatically satisfied and we only need to ensure our Gaussian possibility function upper bounds the ensemble at the other $x_i$'s. Using the invariance of the trace to cyclic permutations allows these constraints to be rewritten as
\begin{equation}
  (x_i - \mu)^{\tr} \Lambda (x_i - \mu) = \Tr \left( C_i \Lambda \right) \leq -2 \log w_i, \qquad 1\leq i \leq N,
\label{variance_ensemble_poss}
\end{equation}
where $C_i  = (x_i - \mu)(x_i - \mu)^{\tr}$; this is more convenient for numerical optimization because $\Tr (C_i \Lambda)$ is linear in $\Lambda$.

\begin{remark}
\label{rem:oneDimensionalRecovery}
In the one-dimensional ($n=1$) case, this optimisation problem can be easily solved for any $N \geq 1$: the $i^{\text{th}}$ constraint in \eqref{variance_ensemble_poss} can be expressed directly for the (scalar) variance as $\sigma^2 \doteq \Sigma \leq -2 \log w_i / (x^i - \mu)^2$, and \eqref{variance_ensemble_poss} reduces to a single constraint:
\begin{equation}
\label{eq:oneDimensionalConstraint}
\sigma^2 \leq \min_{i \in \{1,\dots,N\}} \frac{-2 \log w_i}{(x_i - \mu)^2}.
\end{equation}
In particular, setting $\sigma^2$ to the right hand side of this inequality will maximise the precision, hence solving the optimisation problem \eqref{eq:secondOptimisation}. In this case, if it happens that $w_i = f(x_i)$ with $f$ a Gaussian possibility function, then the associated variance will be recovered exactly, even with $N=1$. Although this result does not generalise easily to higher dimensions, it highlights the potential of this method for recovering the variance from an ensemble, as will be studied later in this section.
\end{remark}

We will denote by $\Lambda^*\big(\{(w_i,x_i)\}_i\big)$ the solution to \eqref{eq:secondOptimisation}, omitting the limits on $i$ for concision. In addition to providing an estimate for the variance of the best-fitting Gaussian, this approach provides a measure of how far from Gaussian the ensemble is. Indeed, a notion of distance can be defined based on the gaps $-\Tr \left( C_i \Lambda \right) - 2 \log w_i$, $i \in \{1,\dots,N\}$.

We will be interested in the relationship between the best Gaussian fit for a given ensemble and the one for a linear and invertible transformation of that ensemble, which motivates the following proposition.

\begin{proposition}
\label{prop:bestFitUnderTransform}
Let $\{(w_i, x_i)\}_i$ be an ensemble such that $w_i = 1$ if and only if $i = 0$, and let $M$ be an invertible linear map on $\mathbb{R}^n$, it holds that
\[
\Lambda^*\big(\{(w_i,x_i)\}_i\big) = M^{\tr}\Lambda^*\big(\{(w_i,M x_i)\}_i\big) M.
\]
\end{proposition}

\begin{proof}
The constraints \eqref{variance_ensemble_poss} for the ensemble $(w_i,M x_i)$ can be rewritten as
\[
\forall i \in \{1,\ldots,N\}: \quad
(Mx_i - M x_0)^{\tr} \Lambda (Mx_i - M x_0) = (x_i - x_0)^{\tr} M^{\tr} \Lambda M (x_i - x_0) \leq -2 \log w_i.
\]
Changing the variable of the optimisation problem from $\Lambda$ to $\tilde{\Lambda} = M^{\tr} \Lambda M$ changes the objective function to $\log | M^{-\tr} \tilde{\Lambda} M^{-1}| = \log |\tilde{\Lambda}| + \text{constant}$. Therefore, the two optimisation problems are equivalent, and their solutions are related as stated.
\end{proof}

From a practical viewpoint, the computational cost for calculating the ensemble's variance via \eqref{eq:secondOptimisation} is greater than that of the probabilistic framework due to the optimisation problem. Yet, \Cref{rem:oneDimensionalRecovery} hints at a potential for an efficient recovery of the true variance in the Gaussian case. This aspect is investigated in \Cref{variance recovery} which displays the average root mean squared error (RMSE) between the true variance and the ensemble's variance based on different sample sizes and for $n \in \{8,16,32\}$, averaged over $1000$ repeats. The true $n \times n$ covariance matrix $\Sigma$ is drawn from the inverse Wishart distribution with $n^2$ degrees of freedom and scale matrix $nI_{n}$, with $I_n$ the identity matrix of dimension $n \times n$, making its reconstruction appropriately challenging. Here, we use a sample $x_i \sim \mathrm{N}(0_n, \Sigma)$ for $i = 1,\dots,N$, with $0_n$ the zero vector of size $n$, to both estimate the probabilistic variance and to be used as the ensemble in \eqref{eq:secondOptimisation} with weights $w_i = \overline{\mathrm{N}}(x_i; 0_n, \Sigma)$. The performance of each case is investigated from the minimum required sample size $N=n$ up to $N=500$. \Cref{variance recovery} shows that the covariance matrix obtained from \eqref{eq:secondOptimisation} indeed converges to the true underlying variance faster than the probabilistic one for small dimensions. In particular, the associated RMSE appears to drop significantly around $N = n^2$.

\begin{figure}[htpb]
\centering
\subfloat[Error between the sample and true variance]{\label{variance recovery}\includegraphics[width=.44\textwidth]{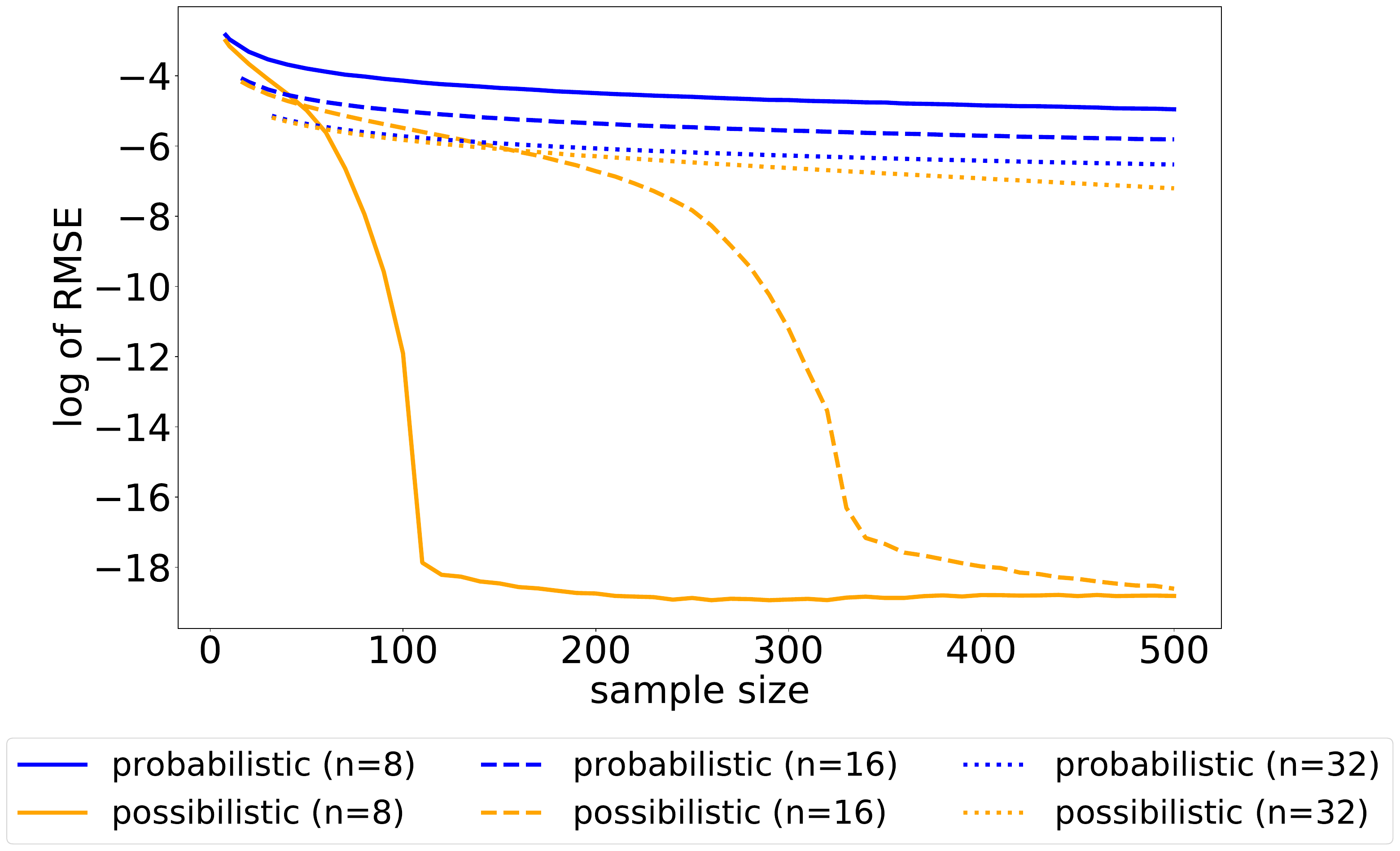}}
\hfill
\subfloat[Computational time of sample variance]{\label{variance processing time}\includegraphics[width=.44\textwidth]{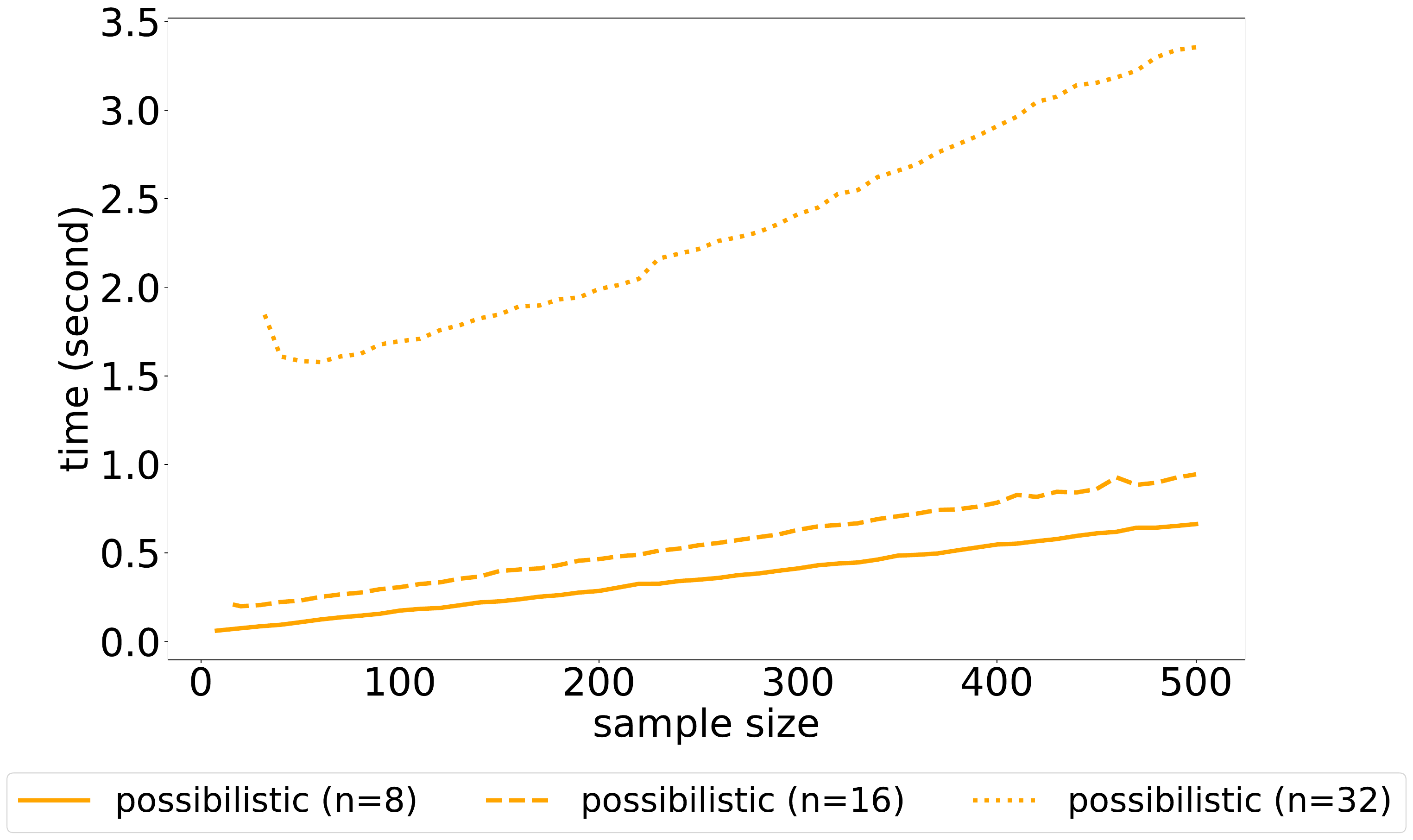}}
\caption{Analysis of the proposed procedure for Gaussian fitting when the underlying possibility function / probability distribution is Gaussian.}
\label{sample covariance}
\end{figure}

The computational cost of the probabilistic approach is very small for the considered range of sample sizes and remains around $\qty{0.1}{\milli\second}$ even for $n=32$ in our experiments, which is likely due to other operations dominating the computational cost in the considered settings. This is in contrast to the computational time of the proposed method, shown in \Cref{variance processing time}, which has a clear linear trend and is orders of magnitude larger than in the probabilistic case.

One advantage of defining an ensemble's variance via \eqref{eq:secondOptimisation} is that the knowledge of dependency---or lack thereof---between components of the states $x \in \mathbb{R}^n$ can be more easily integrated. For instance, if we want to impose conditional independence between some components, i.e., to impose that $\Lambda_{i,j} = 0$ for any indices $i$ and $j \neq i$ in a given set $I$, then we can add a set of constraints in the optimisation problem \eqref{eq:secondOptimisation}. Adding constraints will reduce the precision, i.e., it will yield a Gaussian possibility function that is less informative than the less-constrained problem. This behaviour can be interpreted as follows: in the possibilistic framework, conditional independence can be obtained by sacrificing some information. This is a well-understood trade-off in data assimilation, where forcing components to be uncorrelated, a method known as \emph{localisation}, is usually accompanied by an increase of the variance of the remaining components, a process known as \emph{inflation}. The main difference between the standard inflation and the proposed approach is that inflation usually comes with additional parameters that need to be fine-tuned for different situations, whereas the required amount of precision loss is automatically determined via the optimisation problem \eqref{eq:secondOptimisation} when adding the constraints on conditional independence, with no additional tuning parameters.

Apart from defining the ensemble's expected value and variance, three steps of the algorithm need to be developed: initialisation, prediction, and update. These will be analogous to those of the probabilistic EnKF, but significant changes are required.

\subsection{Initialisation}
\label{Ini poss EnKF}

We consider a sequence of uncertain variables $\bm{x}_0, \bm{x}_1, \dots$, with $\bm{x}_k$ representing the state of the system at time step $k$. We assume that there is prior knowledge about $\bm{x}_0$, encoded into a possibility function $\overline{\mathrm{N}}(\mu_0, \Sigma_0)$. This possibility function is approximated by an ensemble $\{(w^i,\hat{x}^i_0)\}_{i=0}^N$ of $N+1$ weighted particles, defined as in \Cref{sec:ensembleApproximation}. The time index, $0$ in this case, is omitted for the weights as these will in fact remain constant throughout the algorithm, which relies exclusively on transports of the associated particles. The initialisation step is detailed in \Cref{alg: initial step of the p-EnKF}.

\begin{algorithm}
\caption{Initial step of the p-EnKF}
\label{alg: initial step of the p-EnKF}
	\renewcommand{\algorithmicrequire}{\textbf{Input:}}
	\renewcommand{\algorithmicensure}{\textbf{Output:}}
	\begin{algorithmic}[1]
		\REQUIRE The prior Gaussian possibility function $\overline{\mathrm{N}}(\mu_0, \Sigma_0)$
		\ENSURE  The initial ensemble $\{ (w^{i}, \hat{x}_{0}^i ) \}_{i=0}^N$ 
            \FOR{$i \in \{1, \ldots, N \}$}
            \STATE $\hat{x}_{0}^i \sim \mathrm{N}(\mu_0, \Sigma_0)$ \algcom{Draw a particle}
            \STATE $w^i \leftarrow \overline{\mathrm{N}}(\hat{x}_0^i;\mu_0, \Sigma_0)$ \algcom{Calculate the corresponding weight}
            \ENDFOR
            \STATE $\hat{x}_0^0 \leftarrow \mu_0$ \algcom{Add a particle deterministically at the mode}
            \STATE $w^0 \leftarrow 1$
\end{algorithmic}
\end{algorithm}

\subsection{Prediction step}
\label{Pre poss EnKF}

In the probabilistic framework, the standard way of obtaining the predictive ensemble at time step $k$ is to
\begin{enumerate*}[label=\roman*)]
    \item apply the transition model $F_k$ to each particle, and
    \item add a realisation of the transition noise.
\end{enumerate*}
The first part of this process can be used as is, with some advantageous properties: a key result is the fact that $\mathbb{E}^*(F_k(\bm{x}_{k-1})) = F_k(\mathbb{E}^*(\bm{x}_{k-1}))$ with no assumption on either $F_k$ or $\bm{x}_{k-1}$. This allows to perform the prediction without recomputing the expected value, therefore stabilising the estimation through non-linear dynamics. In practice, this means that the particle with index $0$ will always correspond to the expected value, hence its special treatment. In practice, given the ensemble $\{(w^i,\hat{x}_{k-1}^i)\}_{i=0}^N$ at time step $k-1$, we can simply compute the image $\tilde{x}^i_k = F_k(\hat{x}_{k-1}^i)$ of each particle to capture the information at time $k$ in the absence of perturbations.

However, the second part of the standard approach is not appropriate in the possibilistic setting since we aim to model epistemic uncertainty which, in this context, corresponds to deviations between the model and the actual dynamics rather than real perturbations. Deterministic methods can however be adapted and we consider using transport maps to move the points of the ensemble as in \cite{taghvaei2020optimal}. To construct such a map for our setting, we first characterize linear transformations of uncertain variables as follows.

\begin{proposition}
\label{thm: linear mapping}
If $\bm{x}$ is an uncertain variable on $\mathbb{R}^n$ described by $\overline{\mathrm{N}}(\mu, \Sigma )$ and $\bm{z} = A \bm{x}+b $ where $A$ is an $n \times n$ invertible matrix and $b$ is a constant vector of length $n$, then $\bm{z}$ is described by $\overline{\mathrm{N}}(A\mu+b, A\Sigma A^{\tr})$.
\end{proposition}

\begin{proof}
Let $f_{\bm{x}}$ denote the Gaussian possibility function $\overline{\mathrm{N}}(\mu, \Sigma )$. Applying \eqref{change of variable equation}, the possibility function of the uncertain variable $\bm{z}$ is
\begin{align*}
f_{\bm{z}}(z) 
={}&
\sup\{f_{\bm{x}}(x) : x \in \mathbb{R}^n, z = Ax+b \} \\
={}&
\sup\{f_{\bm{x}}(x) : x = A^{-1} (z-b) \} \\
={}&
f_{\bm{x}}(A^{-1} (z-b)) \\
={}&
\overline{\mathrm{N}} \left(A^{-1} (z-b); \mu, \Sigma \right) 
={}
\overline{\mathrm{N}} \big(z; A\mu+b, A \Sigma A^{\tr} \big).
\end{align*}
\end{proof}

According to \Cref{thm: linear mapping}, Gaussianity is preserved under linear and invertible transformations in the possibilistic framework. Furthermore, it is straightforward to obtain a map between two uncertain variables in $\mathbb{R}^n$ described by Gaussian possibility functions:

\begin{proposition}[Mapping between Gaussian possibility functions]\label{thm: Gaussian possibilistic mapping}
If $\bm{x}$ and $\bm{z}$ are two uncertain variables in $\mathbb{R}^n$ described respectively by $\overline{\mathrm{N}}(\mu, \Sigma )$ and $\overline{\mathrm{N}}(\tilde{\mu}, \tilde{\Sigma} )$, then there exists a map $M$ such that $\bm{z} = M(\bm{x})$, which is characterised by
\begin{equation*}
M(x) = \tilde{\mu} + T(x - \mu), \qquad x \in \mathbb{R}^n
\end{equation*}
where $T = \tilde{\Sigma}^{\nicefrac{1}{2}} \big(\Sigma^{\nicefrac{1}{2}} \big)^{-1}$, with the notation $A^{\nicefrac{1}{2}}$ referring to the Cholesky factor of a matrix $A$.
\end{proposition}

\begin{proof}
    First, we rewrite the mapping as follows
\begin{equation*}
M(\bm{x}) ={} \tilde{\mu} + T(\bm{x} - \mu)  ={} T\bm{x} +  \left(\tilde{\mu}   - T\mu \right)
\end{equation*}
As this is a linear and invertible transformation of $\bm{x}$, \Cref{thm: linear mapping} guarantees that $\bm{z} = M(\bm{x})$ is also described by a Gaussian possibility function, defined as
\begin{align*}
f_{\bm{z}}(z) & = \overline{\mathrm{N}}\Big( T\mu +  (\tilde{\mu}   - T\mu ), T\Sigma T^{\tr} \Big) \\
& = \overline{\mathrm{N}}\Big( \tilde{\mu}, \tilde{\Sigma}^{\nicefrac{1}{2}} \big(\Sigma^{\nicefrac{1}{2}} \big)^{-1}  \Big[\Sigma^{\nicefrac{1}{2}} \big(\Sigma^{\nicefrac{1}{2}} \big)^{\tr}  \Big] \big(\Sigma^{\nicefrac{1}{2}} \big)^{-\tr} \big( \tilde{\Sigma}^{\nicefrac{1}{2}} \big)^{\tr} \Big)
 = \overline{\mathrm{N}}\big( \tilde{\mu}, \tilde{\Sigma}\big).
\end{align*}
\end{proof}

Such a transport map can be used in the prediction step of the p-EnKF to add uncertainty from the transition model to the ensemble using deterministic mapping as follows: We fit a Gaussian possibility function $\overline{\mathrm{N}}(\mu_k, \tilde{\Sigma}_k)$ to the ensemble $\{(w^i,\tilde{x}_k^i)\}_{i=0}^N$, then, from the prediction of the possibilistic Kalman filter \cite{houssineau2018smoothing}, we know that an additive Gaussian uncertainty of the form $\overline{\mathrm{N}}(0_n, U_k)$ will yield a Gaussian predictive possibility function with expected value $\mu_k$ and variance $\tilde{\Sigma}_k + U_k$. Using \Cref{thm: Gaussian possibilistic mapping}, we can compute a map $M_k$ from $\overline{\mathrm{N}}(\mu_k, \tilde{\Sigma}_k)$ to $\overline{\mathrm{N}}(\mu_k, \tilde{\Sigma}_k + U_k)$ and apply it to each particle $\tilde{x}^i_k$ to obtain the predicted ensemble at time $k$. The prediction step of the p-EnKF is summarised in \Cref{alg: prediction step of the p-EnKF}. Although we assume the Gaussianity of the ensemble at a slightly earlier stage than the standard EnKF, $U_k$ is typically small compared to $\tilde{\Sigma}_k$ so that the impact of this assumption is expected to be small. In addition, fitting a Gaussian possibility function to the predicted ensemble is unnecessary as the result is known to be $\overline{\mathrm{N}}(\mu_k, \Sigma_k)$, with $\Sigma_k = \tilde{\Sigma}_k + U_k$, based on \Cref{prop:bestFitUnderTransform}. 

\begin{algorithm}
\caption{Prediction step of the p-EnKF at time $k$}
\label{alg: prediction step of the p-EnKF}
	\renewcommand{\algorithmicrequire}{\textbf{Input:}}
	\renewcommand{\algorithmicensure}{\textbf{Output:}}
	\begin{algorithmic}[1]
		\REQUIRE The time $k-1$ posterior ensemble, $\{(w^i,\hat{x}_{k-1}^i)\}_{i=0}^N$.
		\ENSURE  The time $k$ predictive ensemble, $\{(w^i,x_{k}^i)\}_{i=0}^N$, expected value $\mu_k$, and variance $\Sigma_{k}$\\
        \FOR{$i \in \{1, \ldots, N+1 \}$}
            \STATE $\tilde{x}_{k}^i \leftarrow F_k(\hat{x}_{k-1}^i)$ \algcom{Apply dynamics to particles}
        \ENDFOR
        \STATE $\mu_{k} \leftarrow \tilde{x}_{k}^0$
        \STATE $\tilde{\Lambda}_{k} \leftarrow \Lambda^*\big( \{(w^i,\hat{x}_{k-1}^i)\}_i \big)$ \algcom{Compute the precision matrix}
        \STATE $\tilde{\Sigma}_{k} \leftarrow \tilde{\Lambda}_{k}^{-1}$
        \FOR{$i \in \{1, \ldots, N\}$}
            \STATE $T \leftarrow \big(\tilde{\Sigma}_{k} + U_k\big)^{\nicefrac{1}{2}} \tilde{\Lambda}_{k}^{\nicefrac{1}{2}}$ \algcom{Compute matrix for adding uncertainty}
            \STATE $x_{k}^i \leftarrow \mu_k + T(\tilde{x}_{k}^i - \mu_{k})$ \algcom{Transport each particle}
        \ENDFOR
        \STATE $\Sigma_k \leftarrow \tilde{\Sigma}_{k} + U_k$
\end{algorithmic}
\end{algorithm}

\subsection{Update step}
\label{Up poss EnKF}

We first need to specify how the observation equation \eqref{eq:SSMobservation} will be adapted to the considered context. Since perturbations in sensors are often stochastic in nature, we continue to model the error in the observation with a random variable, that is
\[
Y_k = H_k(\bm{x}_k) + \varepsilon_k,
\]
with $\varepsilon_k \sim \mathrm{N}(0,V_k)$ as before. The mechanism to update the information on $\bm{x}_k$ accordingly is provided by \eqref{eq:updateByRandomObs}. In some situations, the main source of uncertainty in the observation will be of an epistemic nature, yet, if the corresponding model errors are described by $\overline{\mathrm{N}}(0,V_k)$, then the posterior possibility function will be the same; this follows from the likelihood principle since the two associated likelihoods will only differ by a multiplicative constant.

There are several variants of the probabilistic EnKF update step. Here, we follow the principles of the SqrtEnKF as it is well suited to the non-random setting of interest. In fact, we will show that our ensemble can be updated exactly in the same way as in the SqrtEnKF.

From the prediction step, we know that the best fitting Gaussian possibility function for the predicted sample is $\overline{\mathrm{N}}(\mu_k, \Sigma_k)$. As is standard, we consider the deviations $e_k^i = x_k^i - \mu_k$ and we verify that the correct updating formulas for the expected value and deviations are
\begin{align*}
\hat{\mu}_k =& \mu_k + K_k(y_k - H_k\mu_k), &
\textrm{ and }\hat{e}_k^i =& e_k^i - \tilde{K}_k H_ke_k^i,
\end{align*}
where $K_k$ and $\tilde{K}_k$ are as defined in the Kalman filter and SqrtEnKF, respectively. The updated particles obtained by adding the posterior expected value and deviations together are
$$
\hat{x}_k^i = \hat{M}_k(x_k^i) \doteq (I_n - \tilde{K}_kH_k) x_k^i + (\tilde{K}_kH_k\mu_k + K_ky_k-K_kH_k\mu_k),
$$ 
which is a linear transformation of $x_{k}^i$. Defining $\bm{x}_k$ as an uncertain variable described by $\overline{\mathrm{N}}(\mu_k, \Sigma_k)$ and using \Cref{thm: Gaussian possibilistic mapping}, it follows that $\hat{M}_k(\bm{x}_k)$ is described by
\begin{align*}
\hat{f}_k & = \overline{\mathrm{N}} \left((I_n - \tilde{K}_kH_k)\mu_k  + (\tilde{K}_kH_k\mu_k + K_ky_k-K_kH_k\mu_k) , (I_n - \tilde{K}_kH_k)\Sigma_k (I_n - \tilde{K}_kH_k)^{\tr}   \right) \\
& = \overline{\mathrm{N}} \left(\mu_k  + K_k(y_k - H_k\mu_k), (I_n - K_kH_k)\Sigma_k \right),
\end{align*}
which matches the update mechanism of the KF, as required. The map $\hat{M}_k$ is therefore moving the particles in such a way that the best fitting Gaussian for $\{( w^i, \hat{M}_k(x^i_k))\}$ is the posterior of the KF with $\overline{\mathrm{N}}(\mu_k, \Sigma_k)$ as a predicted possibility function. Thus, the update step of the p-EnKF is formally the same as that of the standard SqrtEnKF, as detailed in \Cref{alg: update step of SqrtEnKF}, albeit with a difference in interpretation.

\subsection{Extensions}
\label{sec:extensions}

We finish this section by collecting together some extensions to the p-EnKF, demonstrating its applicability in the context of nonlinear measurement models and providing a number of ways to improve computational efficiency.

\subsubsection{Nonlinear observation models} \label{Nonlinear Observation Model}
In the previous section, we have detailed the p-EnKF with a linear observation model. We now establish that, similarly to the EnKF \cite{frei2013ensemble,tang2014nonlinear}, the p-EnKF could be adapted to nonlinear observation models. There are, in fact, two main approaches to dealing with nonlinearity in the observation model, which we consider in turn.

\paragraph{Model linearisation} The simplest way to handle a nonlinear observation model is to linearise it, i.e., to Taylor expand the observation model around the predictive expected value at time $k$ as $H_k(x_{k}) \approx H_k(\mu_{k}) + J_{H_k}(\mu_{k}) (x_{k}-\mu_{k})$, with $J_{H_k}(\mu_{k})$ the Jacobian matrix of $H_k$ at $\mu_k$. Then, a new observation model can be defined based on the observation matrix $J_{H_k}(\mu_{k})$ and on a non-zero mean $H_k(\mu_{k}) - J_{H_k}(\mu_{k})\mu_{k}$ for the observation noise $\varepsilon_k$. After that, we can follow \Cref{alg: update step of SqrtEnKF} for the update, except that the term $y_k - H_k \mu_k$ is replaced by $y_k + J_{H_k}(\mu_{k})\mu_{k}$ in step $4$. Since linearisation does not depend on the chosen representation of uncertainty, it is equally applicable to the p-EnKF as it is to standard versions of the EnKF. The linearisation method can usually be improved by replacing the term $H_k \Sigma_k H_k^{\tr}$ by the predictive variance of the observation based on the ensemble \cite{tang2014nonlinear}, which is also applicable to the p-EnKF. The term $\Sigma_k H_k^{\tr}$ can usually also be replaced by a ensemble-based approximation, yet, this is not straightforward in the p-EnKF since the precision matrix does not have the same properties as the covariance matrix: to compute the covariance matrix $\Sigma_{\bm{x},H_k(\bm{x})}$ between the uncertain variables $\bm{x}$ and $H_k(\bm{x})$, one must first compute the precision matrix for $(
\bm{x}, H_k(\bm{x}))$, invert it, and then extract the block corresponding to $\Sigma_{\bm{x},H_k(\bm{x})}$. Such an extension of the state is however usual, as described next.

\paragraph{Extending the state space} Another way to deal with a nonlinear observation model is by extending/augmenting the original state with a the corresponding predicted observation, which is then observed linearly. In this case, the state becomes $\bm{z}_k = (\bm{x}_k, H_k(\bm{x}_k) )$ and the extended observation matrix is $\tilde{H}_k = \begin{bmatrix} 0_{m \times n} & I_m \end{bmatrix}$, with $0_{m \times n}$ the $0$ matrix of size $m \times n$. \Cref{alg: update step of SqrtEnKF} can be used as usual once a Gaussian is fitted to this extended state. The posterior ensemble can then be extracted by choosing the first $n$ elements of the extended state for every particle. Care must be taken in practice as the precision matrix of the extended state can be close to singularity due to a strong correlation between the elements. For instance, if one of the components of the state that is observed independently becomes sufficiently well estimated at a given time step, then it might be that the nonlinear observation function is approximately linear from the viewpoint of the ensemble. Yet, this corresponds to cases where linearisation would be appropriate. It therefore appears that a hybrid technique would be the most suitable, with components of the observations being either linearised or included in the state depending on their observed degree of nonlinearity.

\subsubsection{Techniques to improve computational efficiency} \label{Technique to deal with small ensembles}
As with any ensemble-based technique, the accuracy of the p-EnKF depends on ensemble size. As it is typically of interest to use small ensembles for computational reasons, it can be challenging to represent the state of interest adequately. An important aspect is that the p-EnKF requires a minimum sample size equal to the state's dimension plus one due to the computation of the variance; a sample size matching the dimension plus one is sufficient to ensure that the resulting covariance matrix is of full rank providing only that the collection of displacements from the expected value to each of the sample points are linearly independent whereas a smaller sample will lead to a rank-deficient covariance matrix. That a sample size equal to the dimension plus one suffices follows from the fact that a quadratic form bounded away from zero at a number of points separated from the centroid by linearly-independent vectors cannot vanish anywhere. However, the required sample size can be reduced by using one technique in the following section. Despite this constraint on the sample size, some methods can be considered for improving the computational efficiency, we present two of them in what follows.

\paragraph{Conditional Independence} \label{Conditional Independence}
For the p-EnKF, variance is computed via an optimisation problem. Thus, the number of variables in the precision matrix of the state $x \in \mathbb{R}^n$ will be $n(n+1)/2$. However, for many problems, there is a natural structure in the state variable, such as a conditional-independence structure, which can be exploited to reduce this number. Indeed, for sparsely dependent models, it is straightforward to reduce the number of nonzero variables in the precision matrix $\Lambda$ using conditional independence as follows:     The off-diagonal elements $\Lambda_{ij}$ that $i \neq j$ are set to $0$ if the variable $\bm{x}_i$ and $\bm{x}_j$ are to be modelled as conditionally independent. By setting some of the off-diagonal elements in the precision matrix to $0$ during the computation of this matrix by optimisation, the other terms in the precision matrix will automatically adjust to these constraints, offering a systemic way to perform inflation that is tailored to the strength of the dependence being assumed away. However, exploiting this structure to gain in computational efficiency would require optimisation algorithms that are specifically tailored to sparse/band diagonal positive-definite matrices. The design of such algorithms is left for future work.

\paragraph{Nodal Numbering Scheme} \label{Nodal Numbering Scheme}
Apart from reducing the number of variables, the calculation and storage can be done more efficiently by reordering the variables to minimise the bandwidth of nonzero entries in the covariance matrix. One of the methods to achieve that is proposed in the literature is called the nodal numbering scheme \cite{cuthill1969reducing}. This method uses the graph to reorder the state's element so that the nonzero elements will be close to the diagonal element. Moreover, the nodal numbering scheme ensures that the permuted matrix will have a bandwidth no greater than the original matrix, making the computation involved with the precision matrix more efficient, see \cite{cuthill1969reducing} and \cite{FlukeThesis} for more details.

\section{Numerical Experiments}
\label{sec: Results}

Here we show the performance of the p-EnKF using simulated data from two different models: a simple linear model and a modified Lorenz 96 model.

\paragraph{Data Generation} One convenient feature of standard probabilistic modelling is that simulation can be performed exactly according to the model: the variability of scenarios, necessary for a thorough performance assessment, can be obtained directly by sampling from the assumed probability distributions. This is no longer the case with possibility theory since embracing epistemic uncertainty means that sampling is no longer a natural operation. The ideal solution would be to obtain a sufficiently large collection of real datasets for which the ground truth is known, this is however not generally achievable for problems like data assimilation. Instead, we generate our simulated scenarios by sampling from the probability distributions assumed by the probabilistic baselines and align our possibilistic model with these.

\subsection{Linear model} \label{Linear model}
We first consider a linear model since the performance can be clearly compared with the optimal filter, which can be computed in this instance via the KF. A simple linear model is considered so as to generalise easily to arbitrary dimensions. In particular, we consider that the $i$-th component at time $k$, $i \in \{2,\dots,n\}$, only depends on the $i$-th and $(i-1)$-th components at time $k-1$. This means that conditional independence can be imposed to reduce the computational cost of obtaining the precision matrix with a limited information loss. The model can be written as a state-space model \eqref{Gaussian State Space equations}, for $k \in \{1,\dots,100\}$, with the following components:
\begin{enumerate}
\item The initial state $X_0$ is sampled from $\mathrm{N}(0_n, 10I_n)$.
\item The dynamic model is linear: $F_k(X_{k-1}) = F_k X_{k-1}$ with
\[
F_k = \begin{bmatrix}1 &\lambda & 0  & \cdots & 0 \\
0 & 1 &\lambda &  \cdots & 0\\
0 & 0 & 1 &  \cdots & 0\\
\vdots & \vdots & \vdots & \ddots & \vdots \\
0 & 0 & 0 &  \cdots & 1\\
\end{bmatrix},
\]
where $\lambda =0.1$ and the covariance matrix of the dynamical noise $\epsilon_k$ is $U_k = 0.01I_n$.
\item The observation model is also linear, $H_k(X_k) = H_k X_k$, with $H_k = \begin{bmatrix}
I_m  & 0_{m \times (n-m)} 
\end{bmatrix}$ and the covariance matrix of the observation noise $\varepsilon_k$ is $V_k = 0.1I_m$.
\end{enumerate}

Based on all the shared properties between the Gaussian possibility function and the Gaussian distribution, the best way to align our possibilistic model with the assumed probabilistic one is to simply keep the expected value and variance parameters in our possibility functions. In particular, we assume that the initial state $\bm{x}_0$ is described by the Gaussian possibility function $\overline{\mathrm{N}}(0_n, 10I_n)$ and that the errors in the dynamical model are described by $\overline{\mathrm{N}}(0_n, U_k)$.

The parameters of the considered methods are as follows: Unless otherwise stated, the parameter $N$ is set to twice the state's dimension, i.e., $N = 2n$. For a given value of $N$, the actual number of samples for all methods is $N+1$. The parameters of the UKF chosen in this paper are given by $\alpha = 0.25, \kappa =130, \lambda = \alpha^2(n+\kappa) - n $ and $\beta = 2$.

\subsubsection{Performance assessment for various sample sizes and dimensions}
\label{Performance based on different sample sizes and dimensions linear} Since the p-EnKF is an ensemble-based method, it is natural for us to investigate the performance based on different sample sizes and dimensions first. \Cref{linear dimsam} shows the performance of the p-EnKF with no localisation when the state is fully observed ($n=m$), comparing it to the SqrtEnKF and the UKF. The performance is measured in terms of average RMSE over $1000$ realisations, except for $n=64$ where we only consider 50 realisations due to a large computational time (more than 30 minutes per run). We aim to assess the performance after initialisation and thus focus on the estimation of the state $X_{100}$ at the last time step. The RMSE is computed \begin{enumerate*}[label=\roman*)]
    \item for the posterior expected value with respect to (w.r.t.)\ the posterior mean of the KF,
    \item for the posterior expected value w.r.t.\ the true state, and
    \item for the posterior variance w.r.t.\ the posterior variance of the KF.
\end{enumerate*}
The key aspects in \Cref{linear dimsam} are as follows:
\begin{enumerate}
    \item As can be seen in \Cref{linear dimsam errmean}, despite the similarities between the p-EnKF and the SqrtEnKF, the latter improves on the former by at least 4 orders of magnitude in terms of RMSE w.r.t.\ the posterior mean of the KF. Although the SqrtEnKF could be used when $N < n$, which is key in large dimensions, the performance would necessary be lower in this regime. The difference in performance is still visible when considering the RMSE w.r.t.\ the true state, as in \Cref{linear dimsam errstate}, however it is less pronounced due to the unavoidable error caused by the distance between the true state and the optimal estimator given by the mean of the KF.
    \item Despite the fact that the capabilities in terms of variance recovery are almost indistinguishable in \Cref{variance recovery} between the possibilistic and probabilistic approach, \Cref{linear dimsam errvar} shows that the p-EnKF once again largely outperforms the SqrtEnKF in terms of RMSE w.r.t.\ the optimal variance given by the KF, with improvements by at least 4 orders of magnitude throughout once again. This is due to the fact that here, as opposed to \Cref{variance recovery}, the mean of the ensemble also needs to be estimated by the SqrtEnKF whereas the expected value for the p-EnkF is given by the particle with index $0$ and thus does not need to be re-estimated.
\end{enumerate}

The estimates of the UKF are closer to optimality than the ones of the p-EnKF in \Cref{linear dimsam errmean,linear dimsam errvar}; this could be due to the difference in initialisation between the two algorithms, with the UKF placing points deterministically and with the p-EnKF relying on a random sample at the first time step. Yet, this source of randomness in the p-EnKF is not necessary and other initialisation schemes are considered in the next section.

\begin{figure}[htpb]
\centering
\subfloat[Average RMSE w.r.t.\ the posterior mean of the KF.]{\label{linear dimsam errmean}\includegraphics[width=1\textwidth]{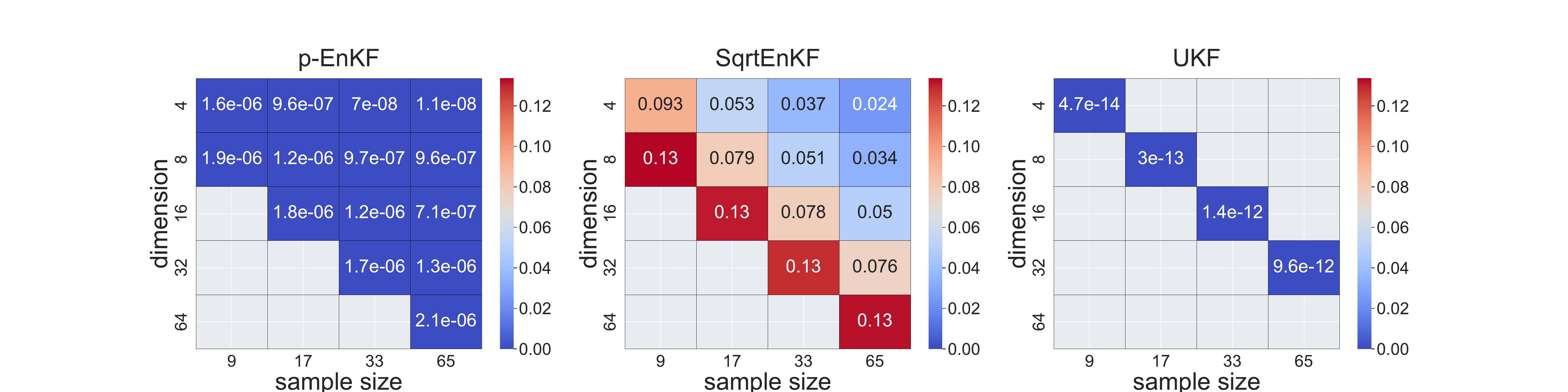}}\\
\subfloat[Average RMSE w.r.t.\ the true state.]{\label{linear dimsam errstate}\includegraphics[width=1\textwidth]{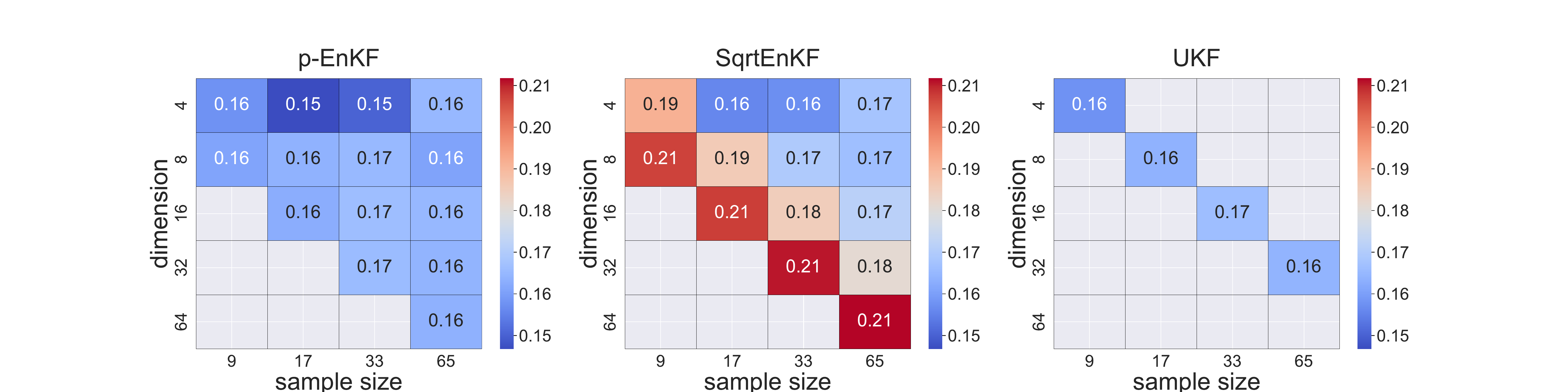}}\\
\subfloat[Average RMSE w.r.t.\ the posterior variance of the KF.]{\label{linear dimsam errvar}\includegraphics[width=1\textwidth]{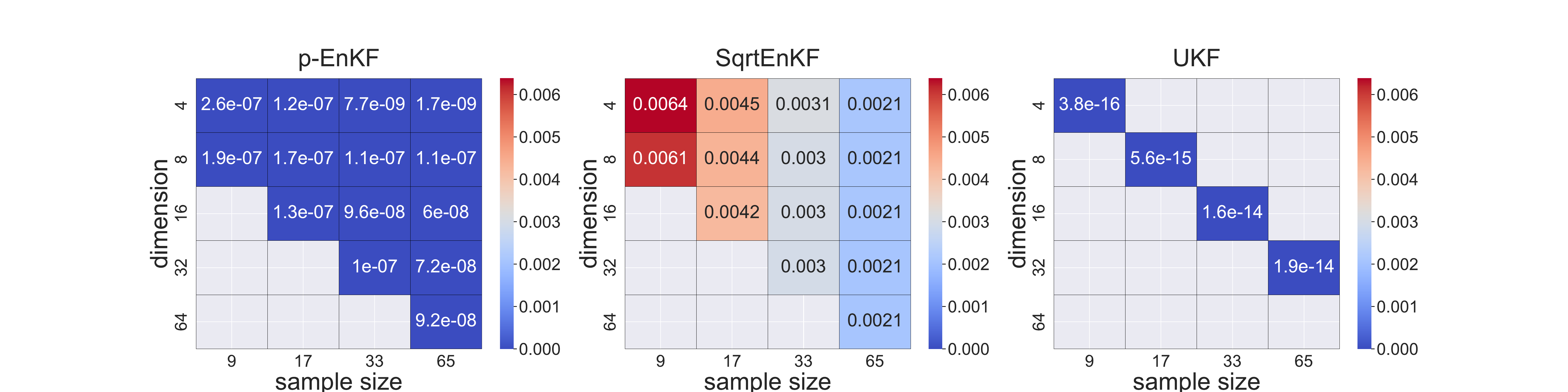}}
\caption{Performance assessment for the fully-observed linear model.}
\label{linear dimsam}
\end{figure}

\subsubsection{Alternatives schemes for particle initialisation} \label{Performance from different algorithms linear}

We now investigate the performance of the p-EnKF with different initialisation schemes inspired by the UKF. When $N = 2n$, we can simply consider exactly the $\sigma$-points of the UKF as initial point for the p-EnKF and we refer to this scheme as ``UKF initialisation''. To allow for setting $N = n$, we also arbitrarily consider only the $\sigma$-points which corresponds to increasing (resp.\ decreasing) one of the components of the mean vector and we refer to this scheme as ``UKF initialisation $+$'' (resp.\ ``UKF initialisation $-$'').

To better highlight the dependency on the initialisation scheme, we consider in this section a partially observed model with $m=1$, i.e., only the first component of the state is observed. This is challenging in general, so only problems of small state dimension are considered. In \Cref{different algorithms with linear dim3and5}, the performance assessment is carried out for $n=3$ and $n=5$, with all RMSEs being averaged over $1000$ repeats. In \Cref{linear dim3 error,linear dim5 error}, all the considered initialisation schemes are compared against the SqrtEnKF and the StEnKF with $2n+1$ samples, with the poor performance of the latter highlighting the difficulty of these inference problems despite the small dimension. There is a slight but consistent improvement in performance when using the UKF initialisation schemes in the p-EnKF, although this improvement vanishes after 30 to 50 time steps, depending on the state dimension. \Cref{linear dim3 range,linear dim5 range} are restricted to the StEnkKF, StEnKF and p-EnKF with UKF initialisation for the sake of legibility, and show that although the interquartile range of the different methods overlap, the p-EnKF performs consistently well across different realisations.

\begin{figure}[htpb]
\centering
\subfloat[Average RMSE w.r.t.\ the true state for $n = 3$]{\label{linear dim3 error}\includegraphics[width=.44\textwidth]{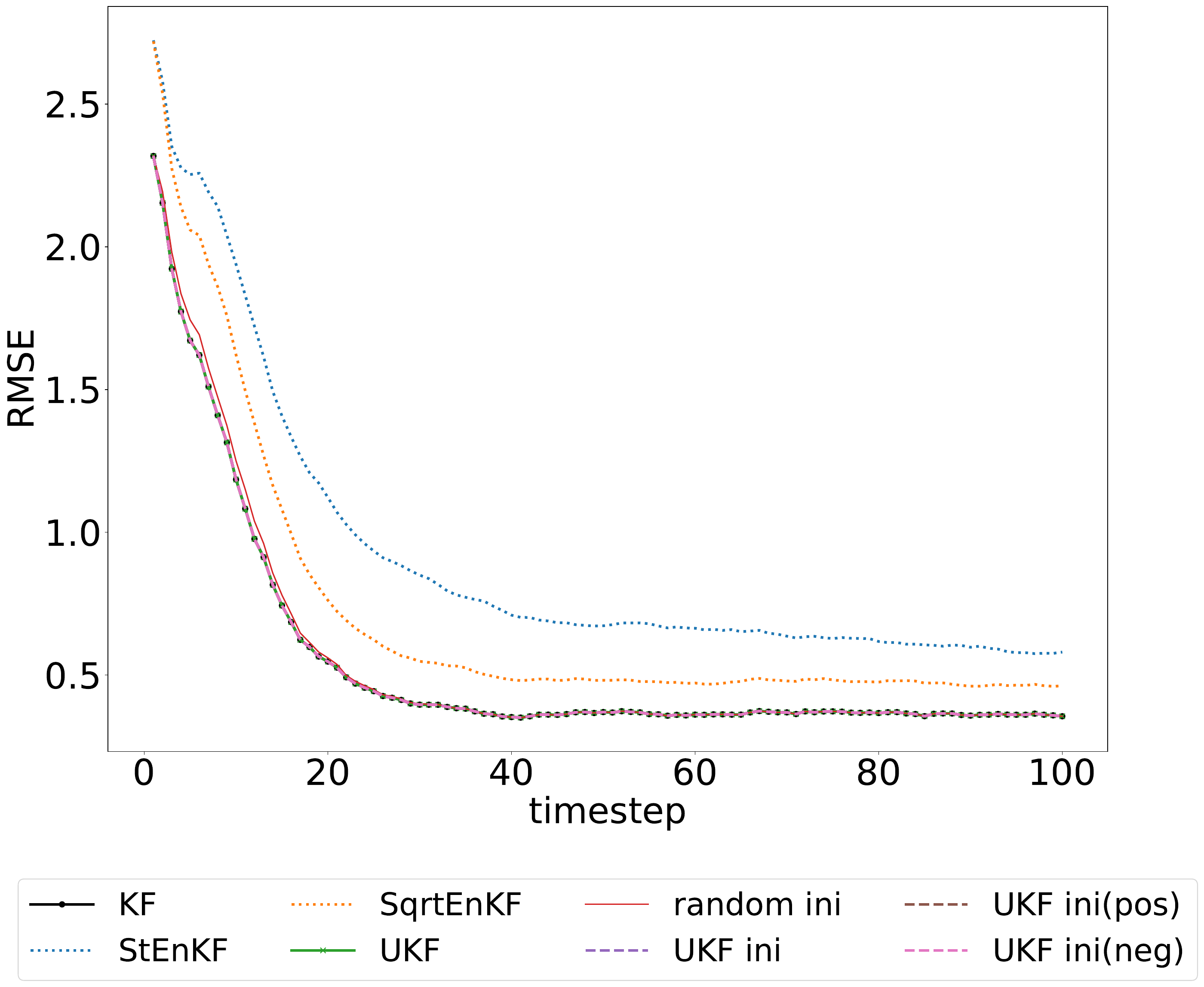}}
\hfill
\subfloat[Average RMSE w.r.t.\ the true state for $n = 5$]{\label{linear dim5 error}\includegraphics[width=.44\textwidth]{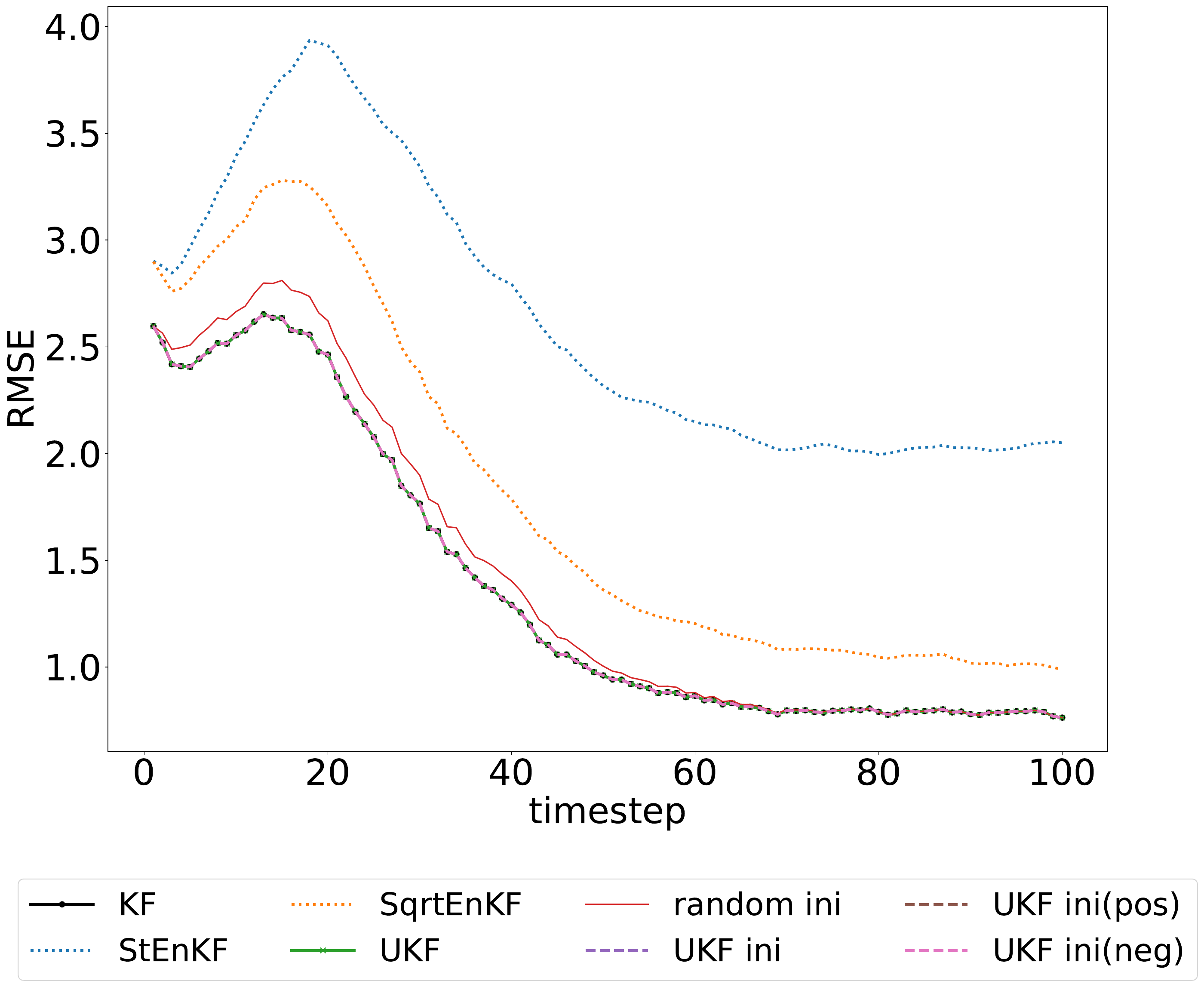}} \\
\subfloat[Average RMSE with the interquartile range for $n = 3$]{\label{linear dim3 range}\includegraphics[width=.44\textwidth]{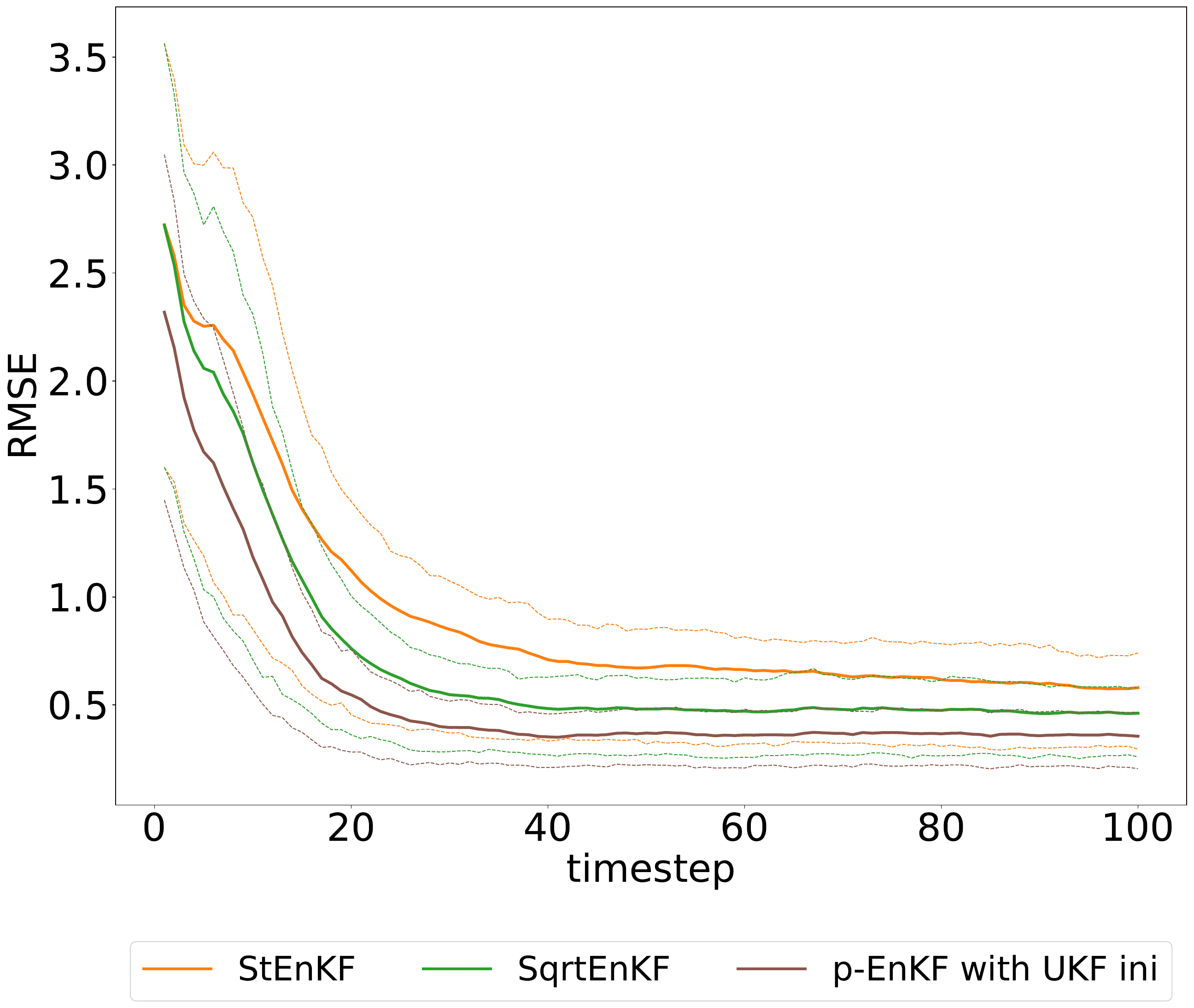}}
\hfill
\subfloat[Average RMSE with the interquartile range for $n = 5$]{\label{linear dim5 range}\includegraphics[width=.44\textwidth]{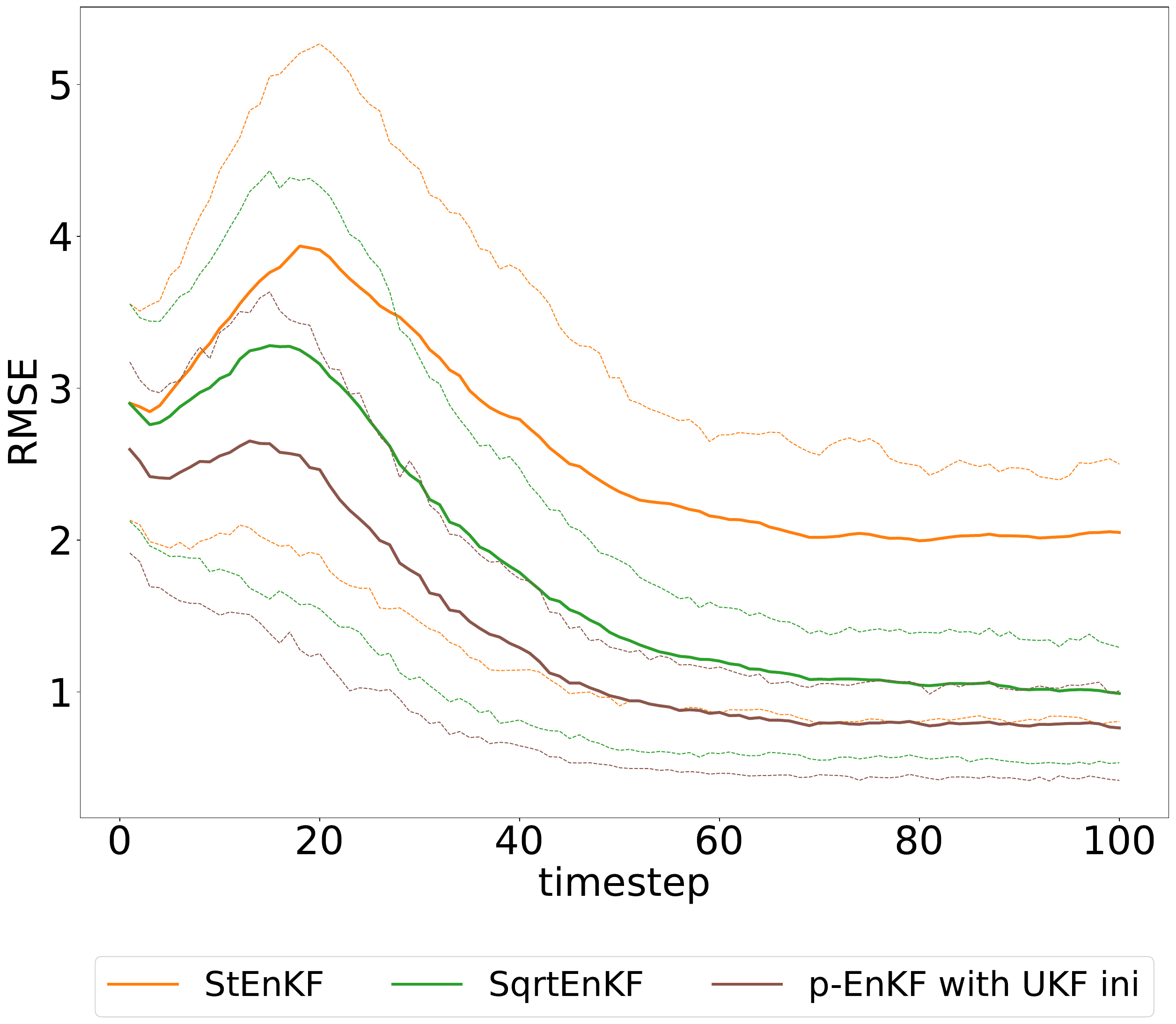}}
\caption{Average RMSE and the error range of different algorithms in the linear model. Left: state's dimension of $3$; Right: state's dimension of $5$, averaged over $1000$ repeats.}
\label{different algorithms with linear dim3and5}
\end{figure}

\subsubsection{Banded vs.\ full precision matrix} \label{Banded VS full precision matrix linear}

To illustrate the capabilities of the p-EnKF to deal with localisation via a systematic form of inflation, we contrast its performance with a full precision matrix against the one with a banded precision matrix with a bandwidth of two, i.e., where all elements except the diagonal and elements adjacent to it are set to zero. We compare these results against the ones obtained with the standard versions of the EnKF for the fully-observed ($n=m=5$) and partially-observed ($n=5,m=1$) cases. The performance is measured with three quantities \begin{enumerate*}[label=\roman*)]
    \item the RMSE of the posterior expected value w.r.t.\ the true state,
    \item the determinant of the posterior variance, and
    \item the Mahalanobis distance between the posterior expected value and the true state.
\end{enumerate*}
The Mahalanobis distance is used to capture how good the estimate is relative to its variance and, hence, assesses the calibration of the algorithms in terms of uncertainty quantification. It is defined as $\sqrt{(x-\mu)^{\tr}\Sigma^{-1}(x-\mu)}$ where $x$ is the true state, $\mu$ is a given posterior expected value, and $\Sigma$ is the posterior variance.

\begin{figure}[htpb]
\centering
\subfloat[Average RMSE w.r.t.\ the true state]{\label{error fullob linear}\includegraphics[width=.44\textwidth]{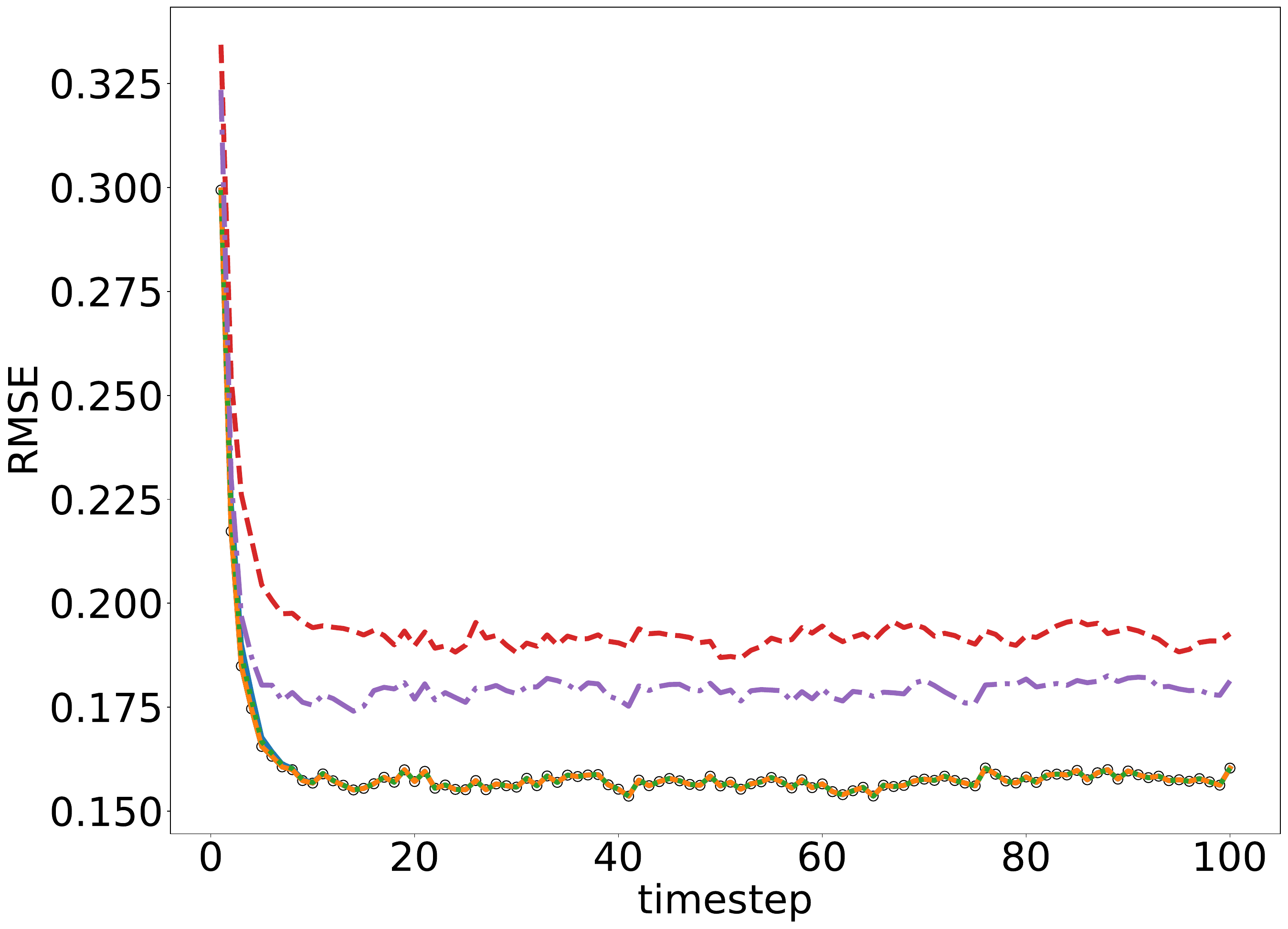}}
\qquad
\subfloat[Average RMSE w.r.t.\ the true state]{\label{error partob linear}\includegraphics[width=.44\textwidth]{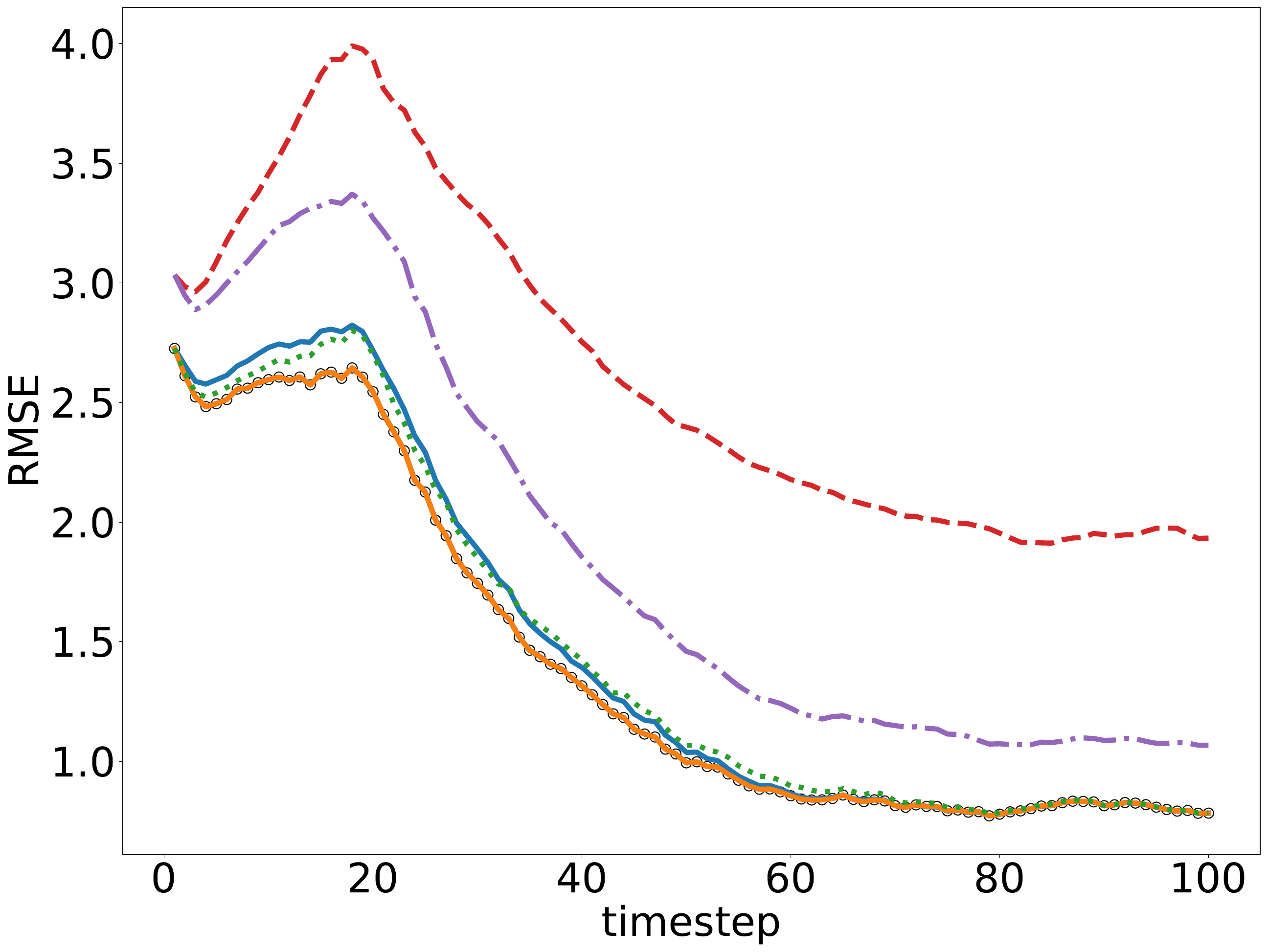}} \\
\subfloat[Estimated $\log(\det(\textrm{variance}))$]{\label{var fullob linear}\includegraphics[width=.44\textwidth]{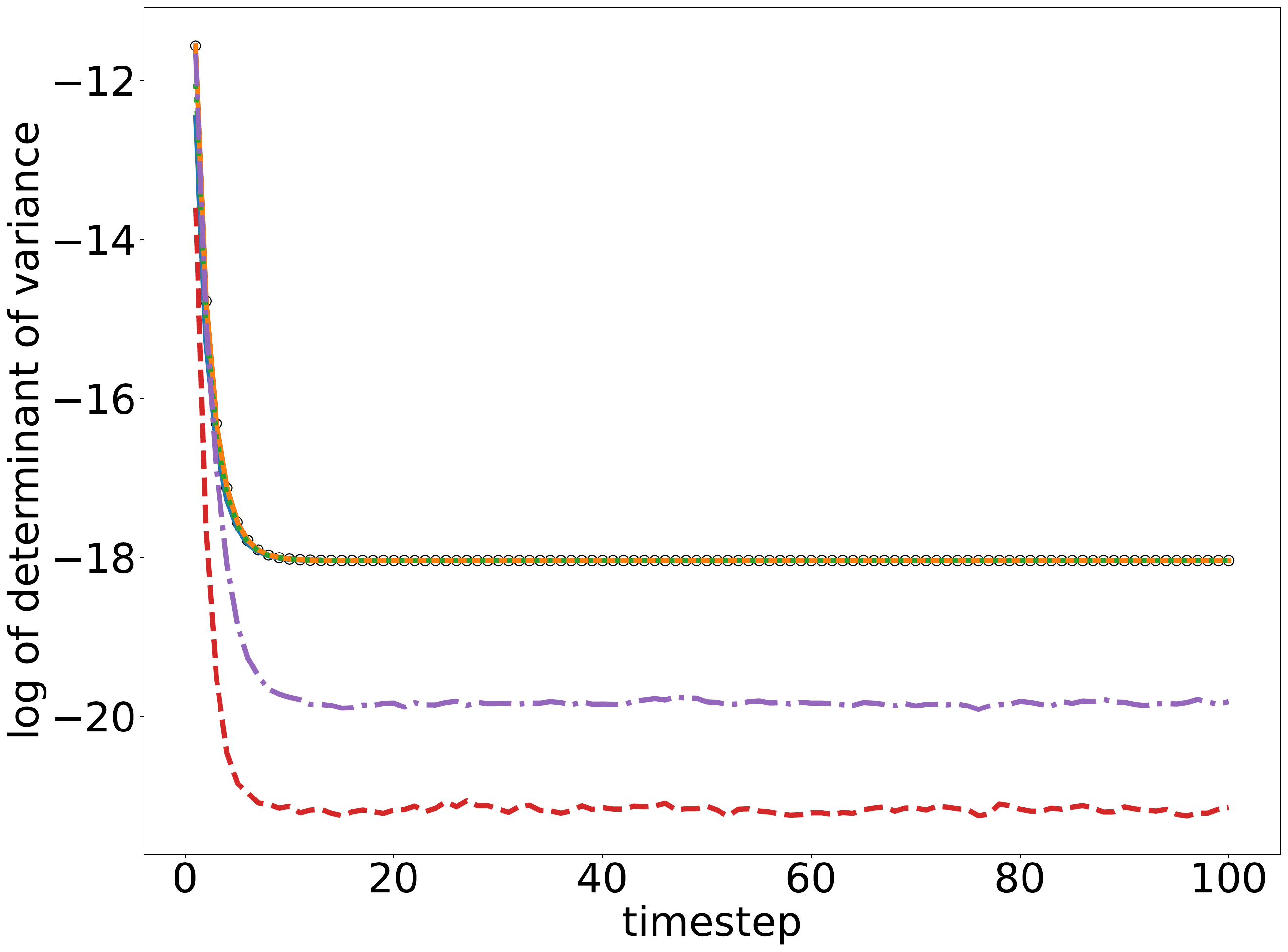}}
\qquad
\subfloat[Estimated $\log(\det(\textrm{variance}))$]{\label{var partob linear}\includegraphics[width=.44\textwidth]{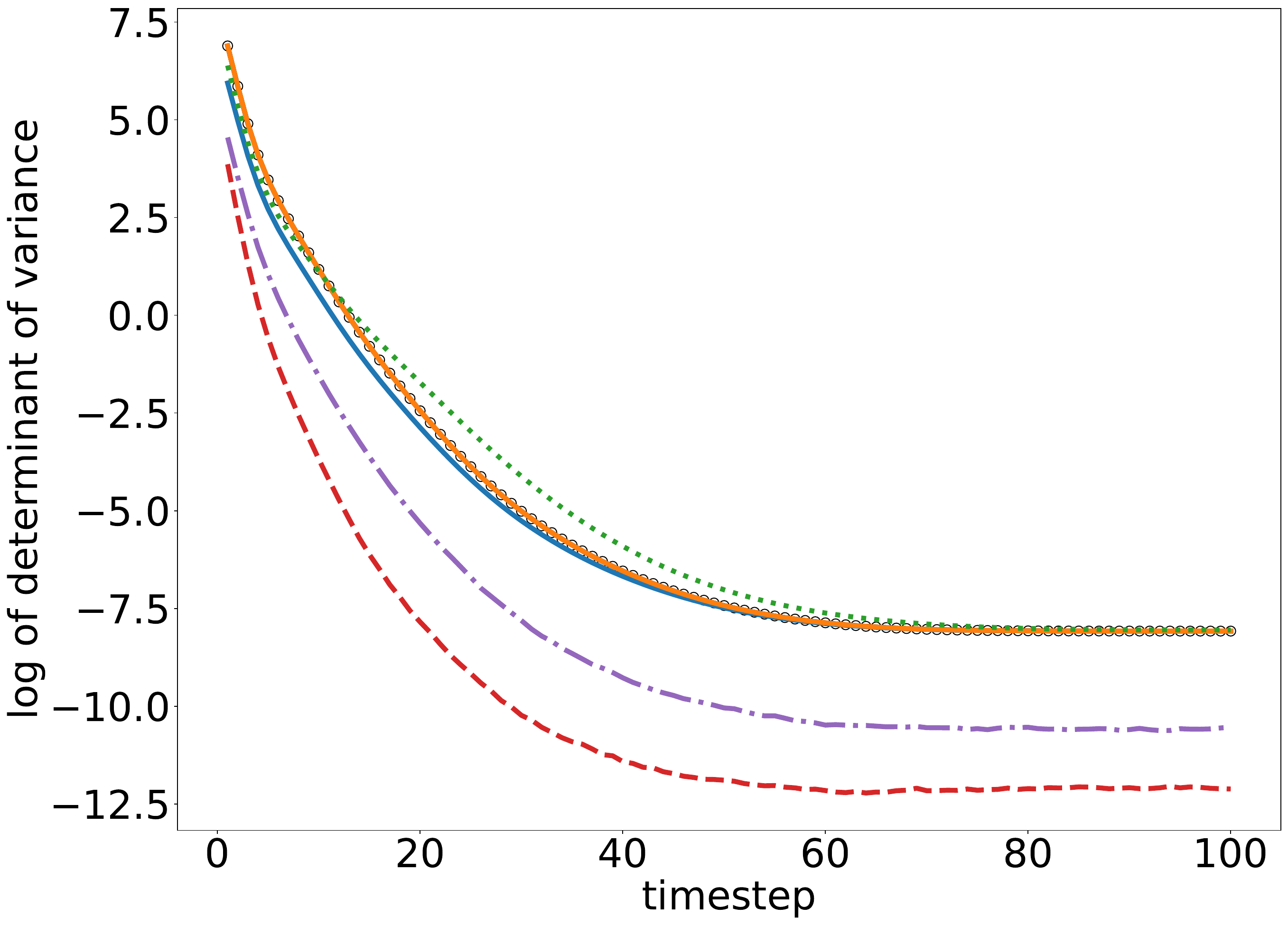}} \\
\subfloat[Mahalanobis distance between the estimate and the true state]{\label{Mdis fullob linear}\includegraphics[width=.44\textwidth]{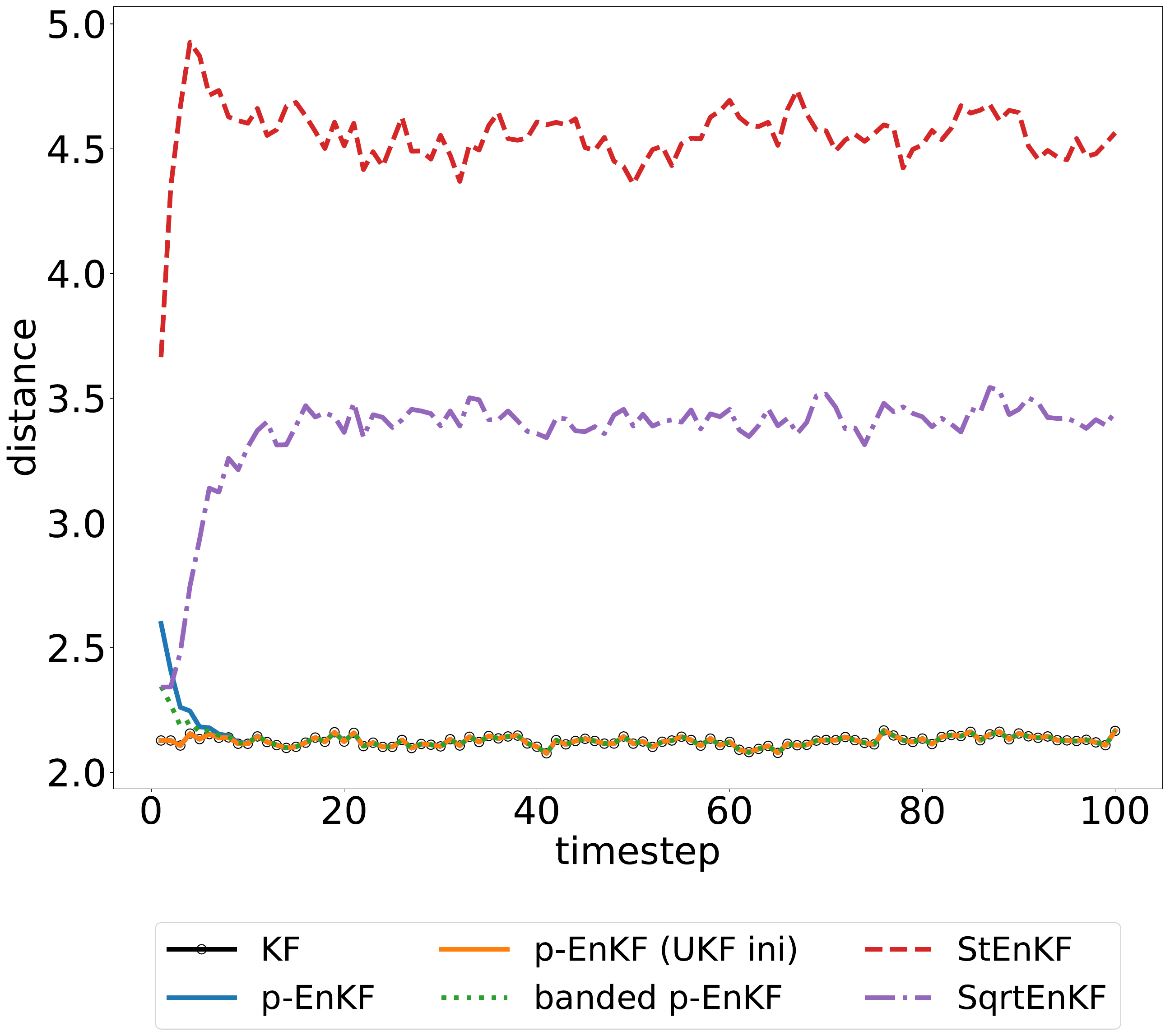}}
\qquad
\subfloat[Mahalanobis distance between the estimate and the true state]{\label{Mdis partob linear}\includegraphics[width=.44\textwidth]{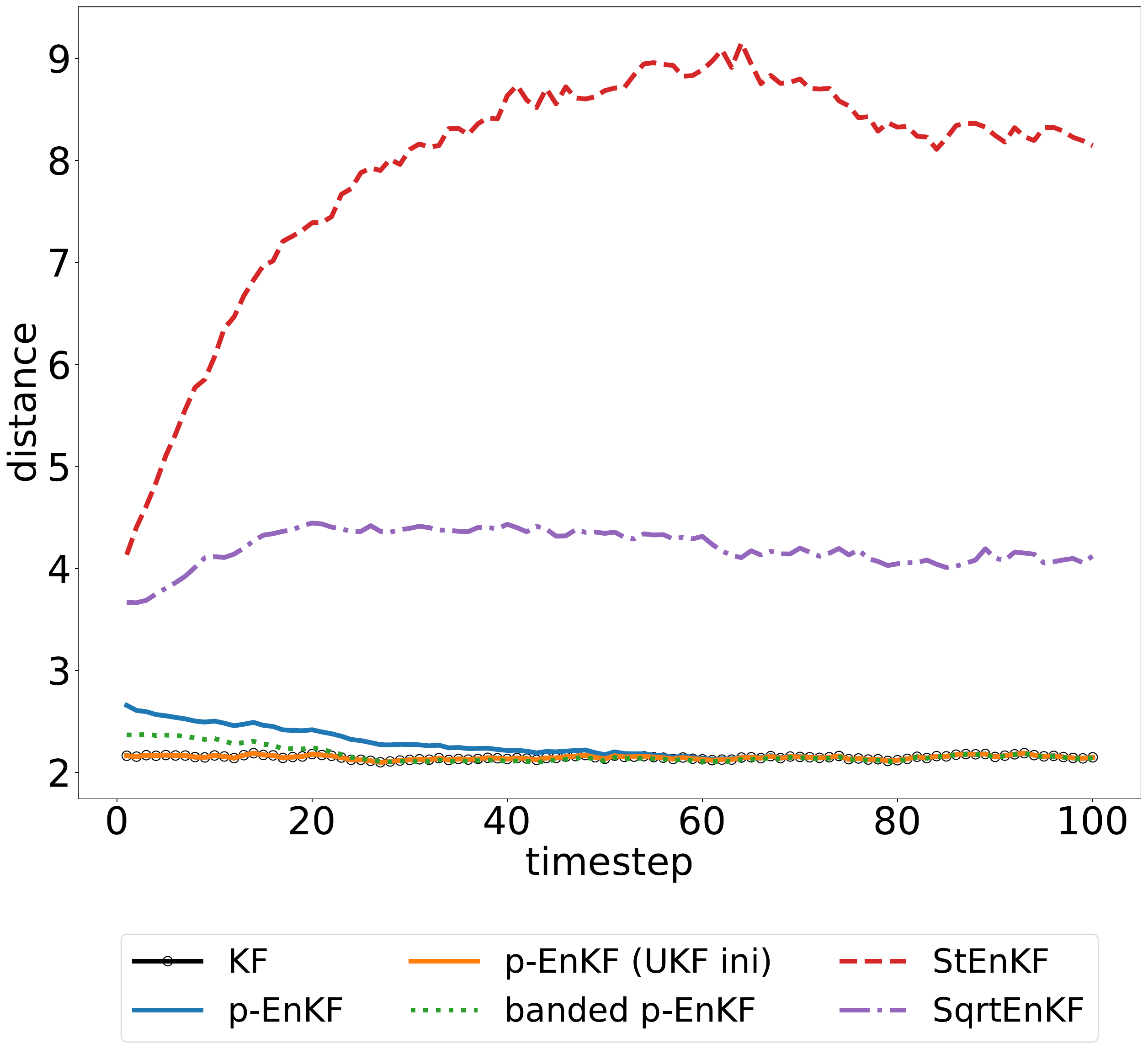}}
\caption{Performance assessment for the linear model with $n=5$ when (left) fully observed, $m=5$ and (right) partially observed, $m=1$, averaged over $1000$ repeats.}
\label{fullbandedlinear}
\end{figure}

\Cref{fullbandedlinear} shows the performance of the p-EnKF with full and banded precision matrix against the same baselines as before. All the algorithms except the KF use 11 samples, i.e., the ensemble size is $N+1$ where $N = 2n$ and the considered performance metrics are averaged over $1000$ repeats. The important aspects in \Cref{fullbandedlinear} are as follows:
\begin{enumerate}
    \item As can be seen in \Cref{var fullob linear,var partob linear}, forcing the precision matrix to be banded has little impact on the precision of the p-EnKF, despite correlations being crucial for a strong performance in the partially observed case (\Cref{var partob linear}). This is confirmed in \Cref{var fullob linear,var partob linear}, where the log-determinant of the posterior variance is mostly unaffected by localisation. A small but noticeable difference can be seen in \Cref{var partob linear}, but the change in variance is in the correct direction: the determinant of the variance was increased by localisation, i.e., some inflation has been automatically applied in order to compensate for the imposed conditional independence. 
    \item The Mahalanobis distance for the two versions of the the p-EnKF is nearly constant and close to the one of the KF for both considered scenario, as seen in \Cref{Mdis fullob linear,Mdis partob linear}. Conversely, the SqrtEnKF  and StEnKF both display large Mahalanobis distances, with the one of the StEnKF even diverging in the partially-observed case; this is due to these algorithms having a larger RMSE than the KF (\Cref{var fullob linear,var partob linear}) but a smaller variance (\Cref{var fullob linear,var partob linear}).
\end{enumerate}

In conclusion, the p-EnKF adopts a nearly optimal behaviour in this linear scenario despite partial observability and localisation. In contrast, the standard versions of the EnKF depart significantly from the behaviour of the KF and tend to be overly optimistic even in the considered small-dimensional inference problems, with well known adverse consequences for downstream tasks for which a reliable quantification of the uncertainty can be crucial.

\subsection{Modified Lorenz 96 type model} \label{Nonlinear model}

In order to show that the strong performance of the p-EnKF observed in the previous section generalises beyond the linear case, we now consider a modified Lorenz 96 (LR96) model, which can be written as a state-space model \eqref{Gaussian State Space equations}, with $k \in \{1,\dots,100\}$, with the following components:
\begin{enumerate}
\item The initial state $X_0$ is sampled from $\mathrm{N}(0_n, 10I_n)$.
\item The deterministic part of the dynamic model is characterised by $x = F_k(x')$ with
\begin{align*}
x_1 & = x'_1 + \left( (x'_2 - c)c - x'_1 + F \right) \Delta t \\
x_2 & = x'_2 + \left( (x'_3 - c)x'_1 - x'_2 + F \right) \Delta t \\
x_i & = x'_i + \left( (x'_{i+1} - x'_{i-2})x'_{i-1} - x'_i + F \right) \Delta t & \textrm{for }  3 & \leq i \leq n-1 \\
x_n & = x'_{n} + \left( (c - x'_{n-2})x'_{n-1} - x'_{n} + F \right) \Delta t 
\end{align*}
where ``$x_i$'' refers to the $i$-th component of $x$ and where $F = 8$, $c = 1$, and $\Delta t = 0.01$.
\item The covariance matrix of the dynamical noise $\epsilon_k$ is $U_k = 0.01I_n$.
\item The observation model is linear, $H_k(X_k) = H_k X_k$, with $H_k = \begin{bmatrix}
I_m  & 0_{m \times (n-m)} 
\end{bmatrix}$ and the covariance matrix of the observation noise $\varepsilon_k$ is $V_k = 0.1I_m$.
\end{enumerate}
Overall, the only difference with the linear model is the transition function $F_k$. Throughout this section, the ensemble size will be set to $2n+1$, that is $N = 2n$, unless specified otherwise. The LR96 model is particularly convenient for performance assessment since the dimension can be easily adjusted.

As before, we start by investigating the performance of the p-EnKF with the full precision matrix based on different sample sizes and dimensions such that all elements of the state are observed $(n=m)$ compared to the SqrtEnKF and the UKF. However, for the nonlinear model, we only examine the performance in terms of RMSE w.r.t.\ the true state at the last step since we can no longer obtain the optimal state estimate from the KF. The RMSE displayed in \Cref{nonlinear dimsam} is averaged over $100$ repeats, except for $n=64$ where we only consider $50$ realisations due to the computational cost. The results shown in \Cref{nonlinear dimsam} are very close to the ones obtained in the linear case, with all three algorithms maintaining a similar level of performance despite the non-linearities.

\begin{figure}[htpb]
    \centering
    \includegraphics[width=\textwidth]{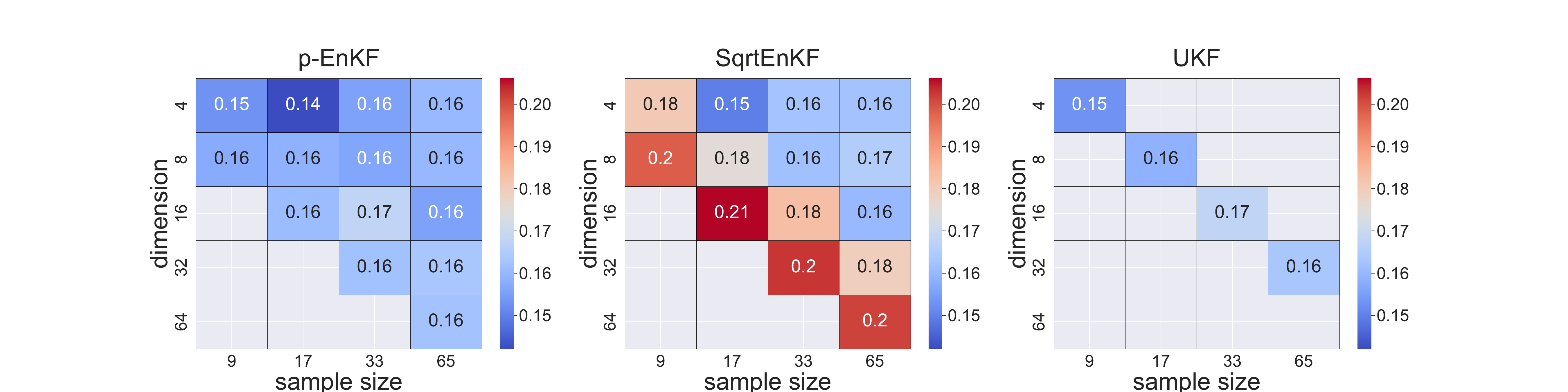}
\caption{Performance of p-EnKF in terms of RMSE w.r.t.\ the true state, compared with the SqrtEnKF and the UKF for the LR96 model}
\label{nonlinear dimsam}
\end{figure}

Because of the similarities between the results with the linear model and the LR96 model, we only highlight where noticeable differences arise. \Cref{error partob nonlinear} shows that, in the partially-observed case with $m=1$ and $n=5$, the SqrtEnKF has a RMSE that is now closer to the StEnKF than to the p-EnKF. Localisation in the p-EnKF still has a very mild effect, although the error increases around time step $70$. This increase in error is however captured by the associated covariance matrix, so that the Mahalanobis distance remains constant throughout the scenario, as required. This behaviour suggests that the correlation increased around time step 70, forcing the p-EnKF with banded matrix to increase the amount of inflation to compensate for the loss of information caused by the imposed conditional independence.

\begin{figure}[htpb]
\centering
\subfloat[RMSE w.r.t.\ the true state]{\label{error partob nonlinear}\includegraphics[width=.32\textwidth,valign=t]{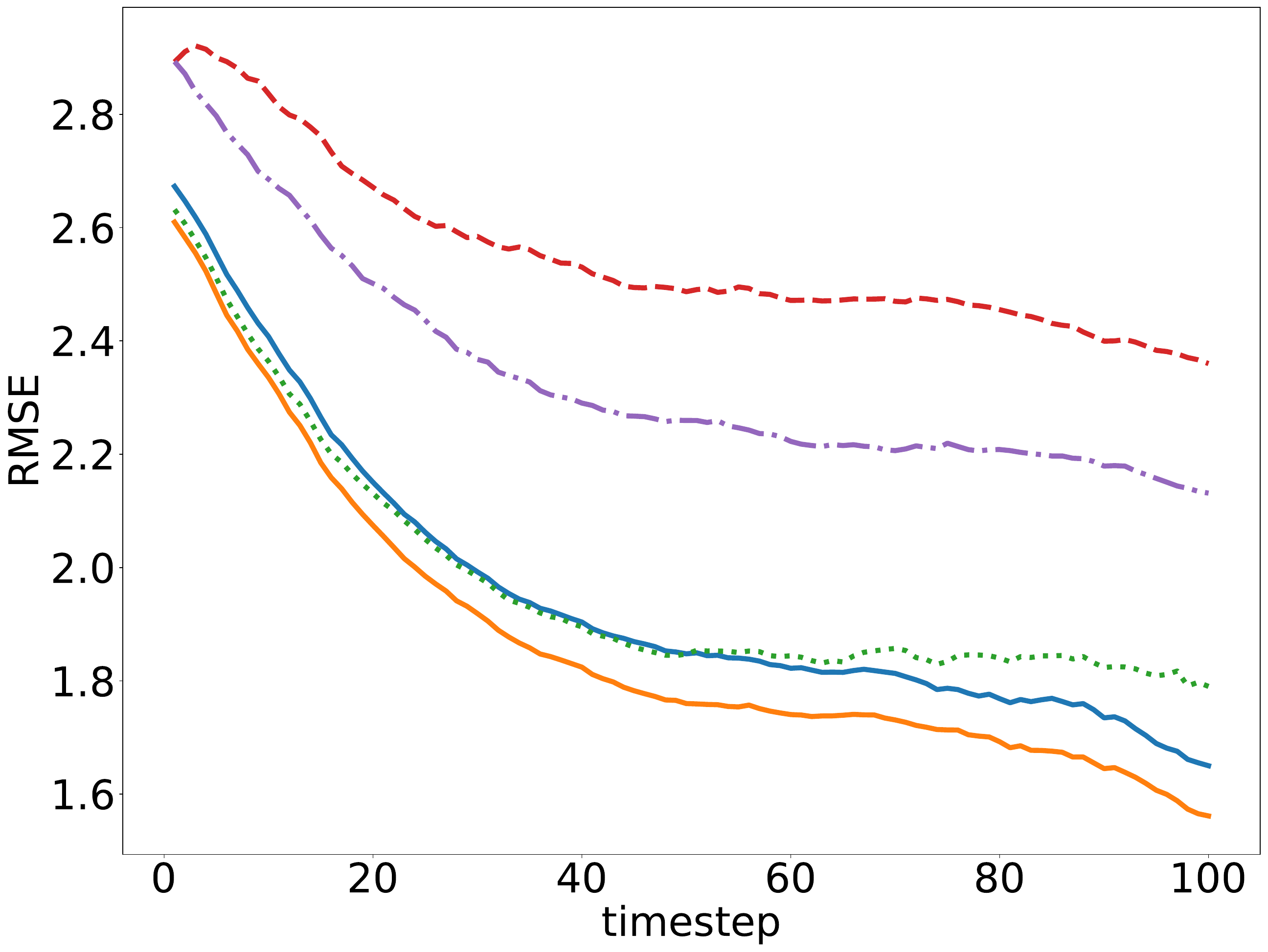}\vphantom{\includegraphics[width=0.32\textwidth,valign=t]{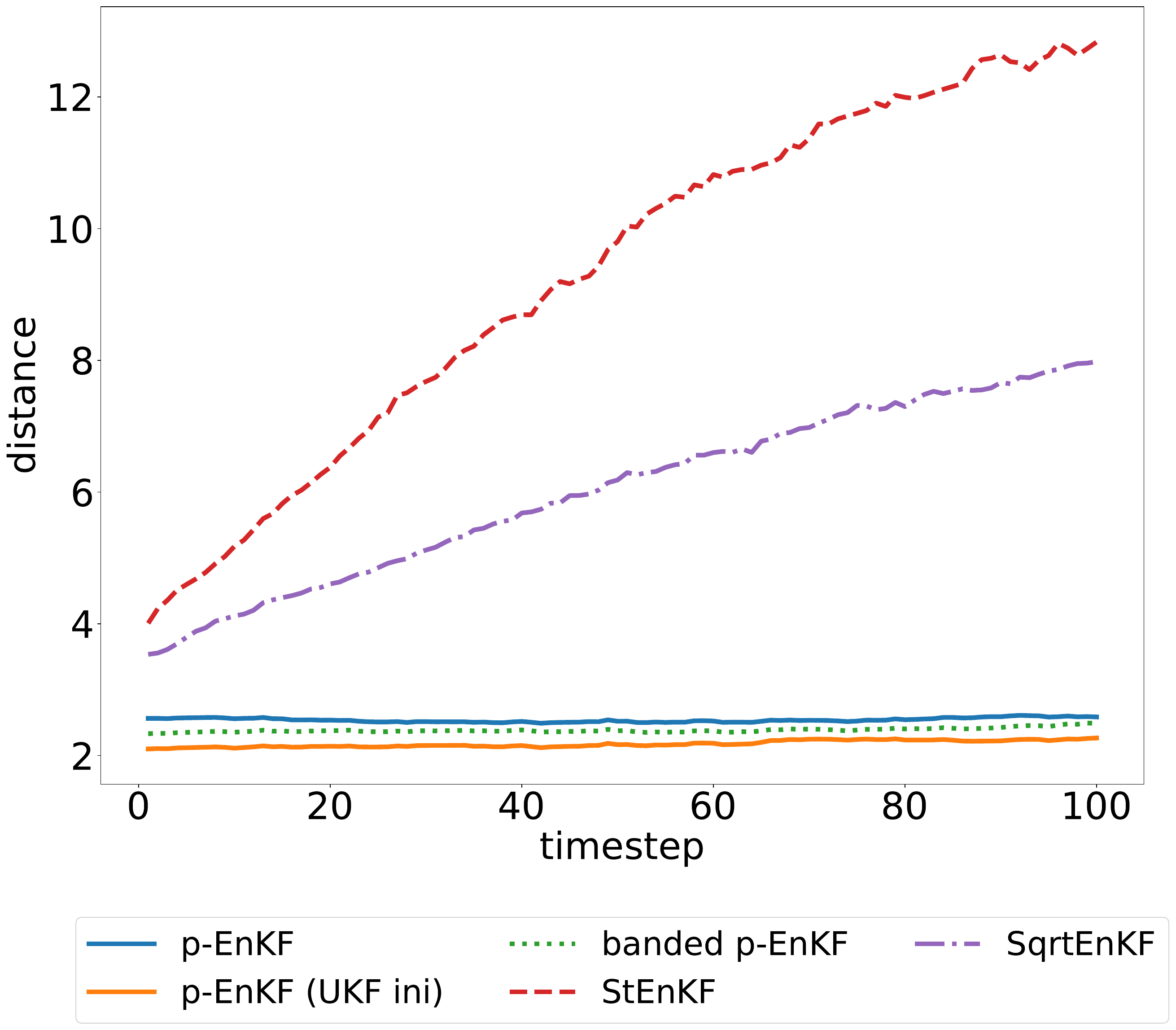}}}
\hfill
\subfloat[Estimated $\log(\det(\textrm{variance}))$]{\label{var partob nonlinear}\includegraphics[width=.32\textwidth,valign=t]{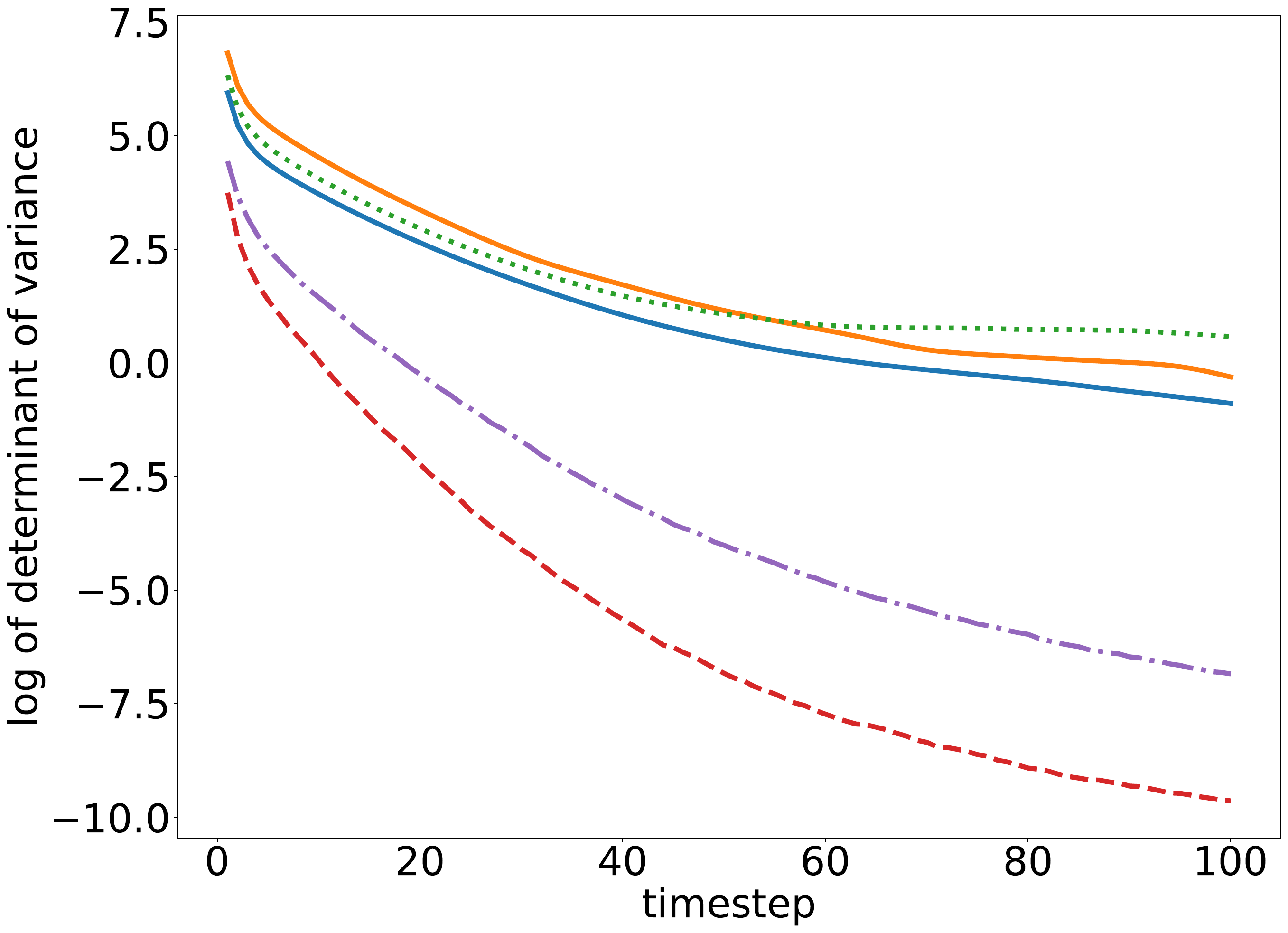}\vphantom{\includegraphics[width=0.32\textwidth,valign=t]{Figures/Mdis_partob_nonlinear.pdf}}} 
\hfill
\subfloat[Mahalanobis distance between the estimate and the true state]{\label{Mdis partob nonlinear}\includegraphics[width=0.32\textwidth,valign=t]{Figures/Mdis_partob_nonlinear.pdf}}
\caption{Performance for a partially-observed LR96 model ($m=1$) when $n = 5$, averaged over $1000$ repeats.}
\label{fullbandednonlinear}
\end{figure}

\section{Conclusion}
\label{sec: Conclusions}

We have introduced the possibilistic ensemble Kalman filter, or p-EnKF, a data assimilation technique treating the state of a state-space model as a fixed quantity about which limited information is available. By using possibility theory to model this form of epistemic uncertainty, we found that much of the intuition behind the standard versions of the EnKF remains valid, with the differences between the theories of possibility and probability leading to key features in the p-EnKF. Specifically, the properties of the expected value and variance in possibility theory appeared to be beneficial for inference problems of small to moderate dimensions, with the p-EnKF closely approximating the Kalman filter in the linear-Gaussian case. These properties also allowed for localisation to be seamlessly applied with no parameter tuning required to compensate for the loss of information incurred by the imposed conditional independence.

In the current version of the p-EnKF, the computation of covariance matrices relies on solving a constrained optimisation problem, which is time-consuming and requires an ensemble size greater than the dimension of the state. Although beyond the scope of this work, lifting these constraints appears to be feasible through the use of specialised optimisation techniques and suitable regularisation. 

There are also several possible avenues for further investigation. We highlight two directions that we think are immediately interesting: although we have found the algorithms performance to be robust to the specification of the initial ensemble, there is a potential for developing systematic approaches to specifying it which might lead to further improved performance with small ensembles; and, an open question is how can the update step be reformulated to operate directly in terms of the precision matrix to avoid the need for matrix inversion, which would be required to facilitate the use of the p-EnKF in high-dimension.

\ack

\bibliography{references}

\begin{thebibliography}{10}

\bibitem{arcucci2020neural}
Rossella Arcucci, Lamya Moutiq, and Yi-Ke Guo.
\newblock Neural assimilation.
\newblock In {\em Computational Science--ICCS 2020: 20th International
  Conference, Amsterdam, The Netherlands, June 3--5, 2020, Proceedings, Part VI
  20}, pages 155--168. Springer, 2020.

\bibitem{asch2016data}
Mark Asch, Marc Bocquet, and Ma{\"e}lle Nodet.
\newblock {\em Data assimilation: methods, algorithms, and applications}.
\newblock SIAM, 2016.

\bibitem{baudrit2008representing}
C{\'e}dric Baudrit, Didier Dubois, and Nathalie Perrot.
\newblock Representing parametric probabilistic models tainted with
  imprecision.
\newblock {\em Fuzzy sets and systems}, 159(15):1913--1928, 2008.

\bibitem{bishop2023robust}
Adrian~N Bishop and Pierre Del~Moral.
\newblock Robust {K}alman and {B}ayesian set-valued filtering and model
  validation for linear stochastic systems.
\newblock {\em SIAM/ASA Journal on Uncertainty Quantification}, 11(2):389--425,
  2023.

\bibitem{burgers1998analysis}
Gerrit Burgers, Peter~Jan Van~Leeuwen, and Geir Evensen.
\newblock Analysis scheme in the ensemble {K}alman filter.
\newblock {\em Monthly weather review}, 126(6):1719--1724, 1998.

\bibitem{chen2021observer}
Zhijin Chen, Branko Ristic, Jeremie Houssineau, and Du~Yong Kim.
\newblock Observer control for bearings-only tracking using possibility
  functions.
\newblock {\em Automatica}, 133:109888, 2021.

\bibitem{cuthill1969reducing}
Elizabeth Cuthill and James McKee.
\newblock Reducing the bandwidth of sparse symmetric matrices.
\newblock In {\em Proceedings of the 1969 24th National Conference}, pages
  157--172, 1969.

\bibitem{de1999conditioning}
Bernard De~Baets, Elena Tsiporkova, and Radko Mesiar.
\newblock Conditioning in possibility theory with strict order norms.
\newblock {\em Fuzzy Sets and Systems}, 106(2):221--229, 1999.

\bibitem{dempster2008dempster}
Arthur~P Dempster.
\newblock The {D}empster--{S}hafer calculus for statisticians.
\newblock {\em International Journal of Approximate Reasoning}, 48(2):365--377,
  2008.

\bibitem{dubois2000possibility}
Didier Dubois, Hung~T Nguyen, and Henri Prade.
\newblock Possibility theory, probability and fuzzy sets misunderstandings,
  bridges and gaps: Misunderstandings, bridges and gaps.
\newblock In {\em Fundamentals of fuzzy sets}, pages 343--438. Springer, 2000.

\bibitem{dubois2015possibility}
Didier Dubois and Henry Prade.
\newblock Possibility theory and its applications: Where do we stand?
\newblock {\em Springer Handbook of Computational Intelligence}, pages 31--60,
  2015.

\bibitem{fearnhead2018particle}
Paul Fearnhead and Hans~R K{\"u}nsch.
\newblock Particle filters and data assimilation.
\newblock {\em Annual Review of Statistics and Its Application}, 5:421--449,
  2018.

\bibitem{fisher1935fiducial}
Ronald~A Fisher.
\newblock The fiducial argument in statistical inference.
\newblock {\em Annals of Eugenics}, 6(4):391--398, 1935.

\bibitem{frei2013ensemble}
Marco Frei.
\newblock {\em Ensemble {K}alman filtering and generalizations}.
\newblock PhD thesis, ETH Zurich, 2013.

\bibitem{guth2021quasi}
Philipp~A Guth, Vesa Kaarnioja, Frances~Y Kuo, Claudia Schillings, and Ian~H
  Sloan.
\newblock A quasi-{M}onte {C}arlo method for optimal control under uncertainty.
\newblock {\em SIAM/ASA Journal on Uncertainty Quantification}, 9(2):354--383,
  2021.

\bibitem{hannig2016generalized}
Jan Hannig, Hari Iyer, Randy~CS Lai, and Thomas~CM Lee.
\newblock Generalized fiducial inference: A review and new results.
\newblock {\em Journal of the American Statistical Association},
  111(515):1346--1361, 2016.

\bibitem{houssineau2018parameter}
Jeremie Houssineau.
\newblock Parameter estimation with a class of outer probability measures.
\newblock {\em arXiv preprint arXiv:1801.00569}, 2018.

\bibitem{houssineau2021linear}
Jeremie Houssineau.
\newblock A linear algorithm for multi-target tracking in the context of
  possibility theory.
\newblock {\em IEEE Transactions on Signal Processing}, 69:2740--2751, 2021.

\bibitem{houssineau2018smoothing}
Jeremie Houssineau and Adrian~N Bishop.
\newblock Smoothing and filtering with a class of outer measures.
\newblock {\em SIAM/ASA Journal on Uncertainty Quantification}, 6(2):845--866,
  2018.

\bibitem{houssineau2019elements}
Jeremie Houssineau, Neil~K Chada, and Emmanuel Delande.
\newblock Elements of asymptotic theory with outer probability measures.
\newblock {\em arXiv preprint arXiv:1908.04331}, 2019.

\bibitem{janjic2011domain}
Tijana Janji{\'c}, Lars Nerger, Alberta Albertella, Jens Schr{\"o}ter, and
  Sergey Skachko.
\newblock On domain localization in ensemble-based {K}alman filter algorithms.
\newblock {\em Monthly Weather Review}, 139(7):2046--2060, 2011.

\bibitem{smc:hmm:Kal60}
R.~Kalman.
\newblock A new approach to linear filtering and prediction problems.
\newblock {\em Journal of Basic Engineering}, 82:35--42, 1960.

\bibitem{kalnay2003atmospheric}
Eugenia Kalnay.
\newblock {\em Atmospheric modeling, data assimilation and predictability}.
\newblock Cambridge University Press, 2003.

\bibitem{katzfuss2016understanding}
Matthias Katzfuss, Jonathan~R Stroud, and Christopher~K Wikle.
\newblock Understanding the ensemble {K}alman filter.
\newblock {\em The American Statistician}, 70(4):350--357, 2016.

\bibitem{FlukeThesis}
Chatchuea Kimchaiwong.
\newblock {\em Ensemble Kalman filtering with an alternative representation of
  uncertainty}.
\newblock {P}h{D} thesis, University of Warwick, 2024.

\bibitem{lemieux2009monte}
Christiane Lemieux.
\newblock {\em Monte {C}arlo and quasi-{M}onte {C}arlo sampling}, volume~20.
\newblock Springer, 2009.

\bibitem{pearl1990reasoning}
Judea Pearl.
\newblock Reasoning with belief functions: An analysis of compatibility.
\newblock {\em International Journal of Approximate Reasoning},
  4(5-6):363--389, 1990.

\bibitem{petrie2008localization}
Ruth~E Petrie.
\newblock Localization in the ensemble {K}alman filter.
\newblock {\em MSc Atmosphere, Ocean and Climate University of Reading}, 2008.

\bibitem{reich2015probabilistic}
Sebastian Reich and Colin Cotter.
\newblock {\em Probabilistic forecasting and {B}ayesian data assimilation}.
\newblock Cambridge University Press, 2015.

\bibitem{robert1999monte}
Christian~P Robert and George Casella.
\newblock {\em Monte Carlo statistical methods}, volume~2.
\newblock Springer, 1999.

\bibitem{sarkka2023bayesian}
Simo S{\"a}rkk{\"a} and Lennart Svensson.
\newblock {\em Bayesian filtering and smoothing}.
\newblock Cambridge University Press, 2023.

\bibitem{taghvaei2020optimal}
Amirhossein Taghvaei and Prashant~G Mehta.
\newblock An optimal transport formulation of the ensemble {K}alman filter.
\newblock {\em IEEE Transactions on Automatic Control}, 66(7):3052--3067, 2020.

\bibitem{tang2014nonlinear}
Youmin Tang, Jaison Ambandan, and Dake Chen.
\newblock Nonlinear measurement function in the ensemble {K}alman filter.
\newblock {\em Advances in Atmospheric Sciences}, 31(3):551--558, 2014.

\bibitem{terejanu2008extended}
Gabriel~A Terejanu et~al.
\newblock Extended {K}alman filter tutorial.
\newblock {\em University at Buffalo}, 2008.

\bibitem{tippett2003ensemble}
Michael~K Tippett, Jeffrey~L Anderson, Craig~H Bishop, Thomas~M Hamill, and
  Jeffrey~S Whitaker.
\newblock Ensemble square root filters.
\newblock {\em Monthly weather review}, 131(7):1485--1490, 2003.

\bibitem{van2015nonlinear}
Peter~Jan Van~Leeuwen, Yuan Cheng, and Sebastian Reich.
\newblock {\em Nonlinear data assimilation for high-dimensional systems: -with
  geophysical applications}.
\newblock Springer, 2015.

\bibitem{wan2000unscented}
Eric~A Wan and Rudolph Van Der~Merwe.
\newblock The unscented {K}alman filter for nonlinear estimation.
\newblock In {\em Proceedings of the IEEE 2000 Adaptive Systems for Signal
  Processing, Communications, and Control Symposium (Cat. No. 00EX373)}, pages
  153--158. IEEE, 2000.

\bibitem{whitaker2002ensemble}
Jeffrey~S Whitaker and Thomas~M Hamill.
\newblock Ensemble data assimilation without perturbed observations.
\newblock {\em Monthly Weather Review}, 130(7):1913--1924, 2002.

\end{thebibliography}

\end{document}